\documentclass[onecolumn,preprint,preprintnumbers,11pt]{article}
\usepackage[utf8]{inputenc}

\usepackage{fullpage} 
\usepackage{mathtools,mathrsfs}
\usepackage{amssymb}
\usepackage{bbm}
\usepackage{amsthm}
\usepackage{dsfont}
\usepackage{mleftright}
\usepackage{hyperref}
\usepackage{graphicx}
\usepackage{tipa}
\usepackage{dsfont}
\DeclareMathAlphabet{\mathpzc}{OT1}{pzc}{m}{it}

\usepackage{upgreek}

\usepackage{algorithm}
\usepackage{algpseudocode}

\usepackage{xcolor}

\newcommand{\bra}[1]{\langle{#1}|}
\newcommand{\ket}[1]{|{#1}\rangle}

\newcommand{\braket}[2]{\langle{#1}|{#2}\rangle}
\newcommand{\ketbra}[2]{|{#1}\rangle\langle{#2}|}
\newcommand{\ketbraz}{\ketbra{0^a}{0^a}}
\newcommand{\eps}{\varepsilon}
\renewcommand{\epsilon}{\varepsilon}
\newcommand{\Ord}[1]{\mathcal{O}\mleft( #1 \mright)}
\newcommand{\tOrd}[1]{\tilde{\mathcal{O}}\mleft( #1 \mright)}

\newcommand{\W}{\ensuremath{\mathcal{W}}}

\newcommand{\id}{\ensuremath{\mathbb{I}}}
\newcommand{\PGM}{\ensuremath{\mathsf{PGM}}}
\newcommand{\tUp}{\ensuremath{\tilde{U}_p}}
\newcommand{\eq}[1]{Eq.~\eqref{#1}}

\newcommand{\Tr}{\mbox{\rm Tr}}

\newcommand{\Pir}{\ensuremath{\Pi_{\rm row}}}
\newcommand{\tPi}{\ensuremath{\tilde{\Pi}}}
\newcommand{\Pic}{\ensuremath{\tilde{\Pi}_{\rm col}}}
\newcommand{\Pin}{\ensuremath{\Pi_{\rm null}}}
\newcommand{\Picn}{\ensuremath{\tilde{\Pi}_{\rm null}}}
\newcommand{\ro}{\ensuremath{\rm row}}

\theoremstyle{plain}

\newtheorem{defn}{Definition}

\newtheorem{theorem}{Theorem}
\newtheorem{lemma}{Lemma}
\newtheorem{corollary}{Corollary}

\def\be{\begin{eqnarray}}
\def\ee{\end{eqnarray}}

\title{Fast algorithm for quantum polar decomposition, pretty-good measurements, and the Procrustes problem}
\author{Yihui Quek and Patrick Rebentrost}
\date{\today}

\begin{document}
\maketitle
\begin{abstract}

The polar decomposition of a matrix is a key element in the quantum linear algebra toolbox. We show that the problem of quantum polar decomposition, recently studied in Lloyd \emph{et al.} \cite{Lloyd2020quantum}, has a simple and concise implementation via the quantum singular value transform (QSVT). We focus on the applications to pretty-good measurements, a close-to-optimal measurement to distinguish quantum states, and the quantum Procrustes problem, the task of learning an optimal unitary mapping between given `input' and `output' quantum states. By transforming the state-preparation unitaries into a {\em block-encoding}, a pre-requisite for QSVT, we develop algorithms for these problems whose gate complexity exhibits a polynomial advantage in the size and condition number of the input compared to alternative approaches for the same problem settings \cite{Lloyd2020quantum, gilyen2020quantum}. For these applications of the polar decomposition, we also obtain an exponential speedup in precision compared to  \cite{Lloyd2020quantum}, as the block-encodings remove the need for the costly density matrix exponentiation step.
We contribute a rigorous analysis of the approach of \cite{Lloyd2020quantum}.
\end{abstract}

\section{Introduction}
It is well-known that every complex number $z$ can be represented in polar form in terms of a magnitude $r\geq 0$ and a phase angle $\theta$, $z= re^{i\theta}$. Analogously, the polar decomposition of a complex matrix $A$ is its representation as the product of a unique positive semi-definite matrix $B$ and a unitary/isometry $U$:
$$A=UB = U\sqrt{A^{\dag}A}.$$
While the polar decomposition is related to the more famous singular value decomposition (SVD), it is well-known in its own right \cite{higham1986computing}. The polar decomposition isometry $U$ is a ubiquitous computational tool due to its `best approximation' property: it is used in spectral clustering \cite{Damle16}, signal processing \cite{bioinformatics,imagecompression}, computer graphics and computational imaging \cite{Shizuo16, Stork17}, and optimization \cite{higham1986computing}. The polar decomposition isometry also occupies a vaunted place in quantum information, having been used to characterize reversible quantum operations \cite{nielsen1998information}, optimal quantum measurements according to various criteria \cite{eldar2002optimalframes, HolPGM78} and the learning of optimal unitary transforms between input and output quantum state pairs, what has been termed the `quantum Procrustes problem' \cite{Lloyd2020quantum, rebentrost2016QuantumSVDNonsparseLowRank}. The last of these is also known as `L\"{o}wdin orthogonalization' \cite{Loewdin1970} in quantum chemistry, where it is a commonly-used method to symmetrically find the closest orthogonal basis of wave functions to a set of non-orthogonal wave functions. 

A prescription for implementing the polar decomposition isometry on a quantum computer would thus unlock its many applications for quantum information processing. This was recently pointed out by Lloyd {\it et al.} ~\cite{Lloyd2020quantum}, which also provided corresponding algorithms based on quantum phase estimation and density matrix exponentiation. In this work, we describe a suite of alternative algorithms for the polar decomposition and its applications, that are based on the recently-introduced quantum singular value transform toolbox~\cite{gilyen2018QSingValTransfArXiv}. 

These algorithms require as input a {\em block-encoding} of $A$. We first assume that a block-encoding of A is readily available, and present a simple algorithm for implementing the polar decomposition, which recasts the singular vector transform of ~\cite{gilyen2018QSingValTransfArXiv}. The remainder of this paper presents applications of the polar decomposition to pretty-good measurements and the Procrustes problem, based on different natural forms of access to A. In the process, we also introduce two novel methods of block-encoding a sum of non-orthogonal rank-1 matrices, that may be of independent interest. These block-encodings allow us to sidestep any computationally expensive quantum phase estimation and density matrix exponentiation steps; as a result, our algorithms are concise and also speed up those of Refs~\cite{Lloyd2020quantum,gilyen2020quantum} polynomially-to-exponentially in various problem parameters. 

\section{Preliminaries}
\subsection{Notation}\label{subsec:notation}
We denote by $[r]$ the set of integers from $0,\ldots, r-1$ (note that we use 0-index notation). When stating algorithmic complexities, we use $\tOrd{}$ to mean that we are using big-O notation that hides poly-logarithmic factors in any of the variables.

\textbf{Matrices:} In the following, let matrix $A \in \mathbbm C^{M \times N}$.  The rowspace of $A$ is denoted as ${\rm row}(A)$ and its kernel is denoted by ${\rm null}(A)$. $\lVert A \rVert$ is the spectral norm of $A$ (which is also its maximum singular value). The polar decomposition is intimately related to the \textbf{singular value decomposition (SVD)}, which is the factorization of $A = W \Sigma V^\dag$ and exists for any $A$. Here $\Sigma \in \mathbbm R^{r \times r}$ is a diagonal matrix whose nonzero entries $\{\sigma_i\}$ are $r\leq \min(M,N)$ positive real numbers, called the {\em singular values} of $A$. $W\in \mathbbm C^{M \times r}$ (resp. $V\in \mathbbm C^{N \times r}$) has $r$ columns $\{\ket{w_i}\}$ (resp. $\{\ket{v_i}\}$) which are the orthonormal eigenvectors of the Hermitian matrix $AA^{\dag}$ (resp. $A^{\dag}A$). We may thus also write $A = \sum_{i\in[r]} \sigma_i \ketbra{w_i}{v_i}$. In this form, it is clear that when $\sigma_i =1\,\forall \, i\,\in[r]$, $A$ is a \textbf{partial isometry}, which is a distance-preserving transformation (isometry) between subspaces (in this case taking ${\rm row}(A)$ to ${\rm row}(A^{\dag})$). This also implies that $W,V$ are partial isometries.

\textbf{Hilbert spaces:} We use $\mathcal{H}$ to denote a finite-dimensional Hilbert space, and $\mathrm{D}(\mathcal{H})$ to represent the set of density matrices, or positive semi-definite operators with trace 1 in $\mathcal{H}$ (i.e. the set of valid quantum states). For clarity, we will sometimes use an integer subscript next to the symbol of a linear operator to denote the number of qubits being acted upon. $\ket{0^a}$ denotes the all-0 basis vector on $a$ qubits. 

\subsection{Polar decomposition}\label{subsec:polar}
\begin{defn}[Polar decomposition]\label{defn:polar_decomp}
Let $A \in \mathbbm C^{M \times N},\, M \geq N$. The polar decomposition of $A$ is a factorization of the form $$A= U B = \tilde B U,$$ where $B$ and $\tilde B$ are unique positive semidefinite Hermitian matrices: $B= (A^{\dag}A)^{1/2}$ and $\tilde{B} = (AA^{\dag})^{1/2}$, and $U$ is a partial isometry. If $A$ is full-rank then $B$ is positive definite and $U$ is unique: $U=A(A^\dagger A)^{-1/2} = (AA^\dagger )^{-1/2} A$. 
\end{defn}
We shall be concerned with enacting $U$, which we call the {\em polar decomposition isometry}. $U$ is unitary if $A$ is square and full-rank. Otherwise $U$ is a partial isometry and is not unique. The following choice of $U$, in terms of $W$, $V$ from the singular value decomposition of $A$, preserves the row and column spaces of $A$, and will be the one that we enact.

\begin{defn}[Polar decomposition isometry $U_{\rm polar}(A)$]\label{defn:polar}
Let $A$ be as in Def.~\ref{defn:polar_decomp}, with singular value decomposition $A = W \Sigma V^\dag$. A polar decomposition isometry for $A$ is 
given by
\begin{equation}\label{eq: polar}
U = U_{\rm polar}(A) := W V^\dag.
\end{equation}
\end{defn}

Many applications of $U_{\rm polar}$ (when it is clear from context, we will drop the `$A$') stem from its ``best approximation" property: for any matrix $A$ its closest matrix with orthonormal columns is the polar factor $U_{\rm polar}(A)$. Higham \cite{Highamblog, higham1986computing} puts it well: ``(This property)...is useful in applications where a matrix that should be unitary turns out not to be so because of errors of various kinds: one simply replaces the matrix by its unitary polar factor." 

The same property is crucial to the polar decomposition's role in the Procrustes problem and pretty-good measurements. In the rest of this paper, we show how to implement $U_{\rm polar}(A)$ on an arbitrary input quantum state, given access assumptions to $A$ natural to both these problem settings. Our algorithms will always assume $A$ is square but can easily be extended to the non-square case by padding, as discussed in the next subsection.

\subsection{Quantum Singular Value Transform}\label{subsec:qsvt}
This subsection summarizes results presented in Ref.~\cite{gilyen2018QSingValTransfArXiv} that are relevant for the present work. 
First we recall the {\em block-encoding} of an operator into a unitary. Let $A \in \mathbb{C}^{N\times N}$ be a complex matrix with $N = 2^n$ and $ \lVert A \rVert \leq \alpha$. In case $A$ is nonsquare or $N$ is not a power of $2$, we pad it with zeros along the necessary dimensions and then the discussion pertains to the padded matrix. 

A unitary matrix $U_A \in \mathbbm C^{2^{(n+a)} \times 2^{(n+a)} }$ is said to be a \textbf{block-encoding} of $A$ if it is of the form
\begin{equation}\label{eq:block-encoding}
U_A=\left[\begin{array}{cc}A / \alpha & \cdot \\ \cdot & \cdot\end{array}\right] \quad \Longleftrightarrow \quad A=\alpha(\bra{0^a}\otimes \mathbbm I_n) U_A(\ket{0^a} \otimes \mathbbm I_n).
\end{equation}
Here, $\alpha$ is called the {\em subnormalization factor} of $A$. $a$ is the appropriate integer that corresponds to the block dimensions of $U_A$ and, in the relevant applications,  is the number of ancilla qubits necessary for the complete algorithm including all subroutines. Often we cannot encode $A$ precisely in the top-left-block of the unitary $U_A$, but it will suffice to encode a good approximation of it, i.e.,
\begin{equation}\label{eq:block-encoding-approx}
\left \lVert A-\alpha(\bra{0^a} \otimes \mathbbm I_n ) U_A(\ket{0^a} \otimes \mathbbm I_n) \right \rVert \leq \delta.
\end{equation}
(c.f. Equation \eqref{eq:block-encoding}). Such a $U_A$ is a $(\alpha, \delta, a)$-block-encoding of $A$.

We now describe how block-encodings are used in practice. Consider the $n$-qubit input state $|\psi\rangle$ and $a$ ancilla qubits in the state $\ket {0^a}$ and apply the block-encoding $U_A$ to it. From Eq.~\eqref{eq:block-encoding}), the output state is
\begin{align}\label{eq:output}
U_A  \ket{0^a} \otimes \ket \psi &= \ket{0^a} \otimes A\ket{\psi}/\alpha  + \ket \perp,
\end{align}
where the orthogonal part is $\ket{\perp} := \displaystyle \sum_{i \in \{0,1\}^a,\\ i \neq 0^a} \ket{i} \otimes \ket{\phi_i}$ with garbage states $\ket{\phi_i}$ defined by the unitary $U_A$. In effect Eq.~\eqref{eq:output} means that we have implemented the desired transformation $A$ on the input state, storing the resulting state in a subspace marked by $\ket{0^a}$. In practice we often {\em post-select} on the $\ket{0^a}$ outcome via amplitude amplification (which we also use later).

Naturally, we may define projectors into the subspaces given by the various possible settings of the ancilla qubits. The most important of these is the $\ket{0^a}\bra{0^a}$ subspace; accordingly define
\be\label{eq:projectors}
\tilde \Pi = \Pi := \ket{0^a} \bra{0^a} \otimes \mathbbm I_{n},
\ee
and notice that that $\tilde \Pi U_A\Pi = \ket {0^a}\bra{0^a} \otimes  A/\alpha$. Note that, like in the original Ref.~\cite{gilyen2018QSingValTransfArXiv}, we use  $\tilde \Pi$ for left projectors
and $\Pi$ for right projectors. As $A$ is square, $\tilde{\Pi}$ is actually the same as $\Pi$, but this notation becomes less redundant in the next definition, where we introduce projectors into subspaces of the row/column space of $A$. For the remainder of this paper, $W, V$ will refer to the SVD factors of $\ket{0^a}\bra{0^a} \otimes A$, not $A$ itself.

 \begin{defn}[Projectors] \label{defn:proj}
Let $\Pi, \tilde{\Pi}$ be as in Eq.~\eqref{eq:projectors}, let $U_A$ be a $(1,0,a)$ block-encoding of $A$, and let $\tilde{\Pi} U_A \Pi = W\Sigma V^{\dag}$, with $\sigma_{\min}(A)$ being the minimum singular value of $A$.
Let $\Sigma_{\geq \delta}$ be the matrix obtained by setting to $1$ all values in $\Sigma$ at least $ \delta$ and to $0$ everything else. Define:
\begin{enumerate}
    \item The right and left singular value projectors, $\Pi_{\geq \delta} := \Pi V \Sigma_{\geq \delta} V^{\dag} \Pi$ and $\tilde{\Pi}_{\geq \delta} := \tilde{\Pi} W\Sigma_{\geq \delta} W^{\dag} \tilde \Pi$ respectively.
    \item The projectors onto the rowspace of $A$ and column space of $A$, $\Pir := \Pi_{\geq \sigma_{\min}(A)}$ and $\Pic := \tilde \Pi_{\geq \sigma_{\min}(A)}$ respectively.
    \item The projectors onto the right and left null space of $A$, $\Pin :=  \id - \Pir$ and $\Picn := \id - \Pic$ respectively.
\end{enumerate}
 \end{defn} 
In this work, we require a Theorem from Ref.~\cite{gilyen2018QSingValTransfArXiv}, which will be the core subroutine used to implement the polar decomposition isometry. 
\begin{theorem}[Singular vector transformation, Theorem 26 of \cite{gilyen2018QSingValTransfArXiv}] \label{thm:SVT}
Let $\eps, \delta>0$ and given any projectors $\tilde{\Pi}, \Pi$ and unitary $U$, let the singular value decomposition of the projected unitary $\tilde{\Pi} U \Pi$ be $\tilde{\Pi} U \Pi = W\Sigma V^{\dag}$. Then there is a unitary $U_\Phi$ with $m=\mathcal{O}\left(\frac{\log (1 / \varepsilon)}{\delta}\right)$ such that 
 \begin{equation}\label{eq:SVT_guarantee}
     \left\lVert {\tilde\Pi}_{\geq \delta} U_{\Phi} \Pi_{\geq \delta}-{\tilde\Pi}_{\geq \delta}\left(W V^{\dagger}\right) \Pi_{\geq \delta} \right\rVert \leq \varepsilon.
 \end{equation} 
 Moreover, $U_{\Phi}$ can be implemented using a single ancilla qubit with $m$ uses of $U$ and $U^{\dagger}$, $m$ uses of $C_{\Pi}NOT$ and $C_{\tPi}NOT$, and $m$ single qubit gates.
 \end{theorem}

\subsection{Quantum polar decomposition}\label{subsec:faster_polar}

The singular vector transformation (Theorem \ref{thm:SVT}) forms a building block for the rest of the algorithms in this paper. If an exact\footnote{For the sake of simplicity, we assume the block-encoding provided is exact; however the techniques are robust to inexact block-encodings as shown in Lemma 22 of \cite{gilyen2018QSingValTransfArXiv}.} block-encoding of $A$ is readily available, transforming its singular vectors prepares a good approximation of $U_{\rm polar}(A)$. This unitary can then be applied on any input state $\ket{\psi}$, as stated in the following Corollary, proven in Appendix~\ref{app:proofs_fasterpolar}.

\begin{corollary}[Polar decomposition isometry\label{corr:polar}]
Given an input state $\ket \psi$ and from Theorem 1 with $U\gets U_A$, $\epsilon \gets \eps^2$, $\delta \leq \sigma_{\min}(A)$, $\ket {\psi_{\rm out}} := U_\Phi \ket{0^a} \ket \psi$ satisfies
\begin{equation}\label{eqCor1}
 \lVert \tilde \Pi \ket{\psi_{\rm out}} - \ket{0^a} U_{\rm polar} \ket{\psi}  \rVert \leq \eps^2. 
 \end{equation}
If additionally $\ket{\psi} \in {\rm row}(A)$, then
 \begin{equation}\label{eqCor1_2}
 \lVert \ket{\psi_{\rm out}} - \ket{0^a} U_{\rm polar} \ket{\psi}  \rVert \leq \eps. 
 \end{equation}
Moreover, $U_{\Phi}$ can be implemented using $\Ord{\log(1/\varepsilon) \kappa}$ uses of $U_A$ and $U_A^{\dag}$ and $\Ord{\log(1/\varepsilon) \kappa}$ uses of $C_{\ketbraz}NOT$s and single-qubit gates.
\end{corollary}
Hence, the transformed part of the state (which is $U_{\rm polar}$ applied to $\Pir \ket{\psi}$) lies in the $\ketbra{0^a}{0^a}$ subspace of the output state. However, if we desire that the entire output state is proportional to $U_{\rm polar}\ket{\psi}$, we should amplify this subspace, as we now show. \begin{algorithm}[H]
\textbf{Input:} Input state $\ket{\psi}$, $U_A$ a $(1,0,a)$-block encoding of $A$, accuracy parameter $\eps$.\\
\textbf{Output:} $\ket{\psi_{\rm out}}$.
 		\vspace{5pt}
 \begin{algorithmic}[1]
 \State \textbf{Initialization:} initialize the state $\ket{0^a}\otimes \ket{\psi}$.
 \State Prepare $U_{\Phi}$ with singular vector transformation Theorem \ref{thm:SVT}. 
 \State (Optional) Prepare $U_{\rm AA}$, the unitary that amplifies the $\ket{0^a}\bra{0^a}$ subspace using robust oblivious amplitude amplification on the unitary $U_{\Phi}$. 
 \State  $\ket{\psi_{\rm out}} \gets U_{\rm AA} \ket{0^a}\otimes \ket{\psi}$.
 \end{algorithmic}
 \caption{Polar decomposition with amplification}
 \label{algo:polar}
 \end{algorithm}
The guarantees of Algorithm \ref{algo:polar} and the appropriate parameter choices in each subroutine are given in Theorem \ref{thm:polar},
proved in Appendix \ref{app:proofs_fasterpolar}.
\begin{theorem}[Polar decomposition isometry with amplification\label{thm:polar}]
Let $\kappa:= 1/{\sigma_{\min}(A)}$ and $\Pi$ be as in Eq.~\eqref{eq:projectors}.
Given an $n$-qubit quantum state $\ket{\psi}$ such that 
$\lVert \Pir \ket{0^a}\ket{\psi} \rVert_2 =c$, 
Algorithm \ref{algo:polar} prepares an $n+a$-qubit state $\ket{\psi_{\rm out}}$ that satisfies
\be\label{eq:statecloseness}
 \left\lVert \ket{\psi_{\rm out}}  -  \frac{1}{\lVert WV^{\dag} \ket{0^a}\ket{\psi} \rVert_2} WV^{\dag}\ket{0^a} \ket{\psi} \right\rVert \leq \varepsilon,
 \ee
with $\tOrd{\frac{\kappa}{c} \log(1/\varepsilon) }$ uses each of $U_A$ and $U_A^{\dag}$, $C_{\Pi}NOT$ gates and single-qubit gates. 
\end{theorem}

\textbf{Comparison to prior work.} Ref.~\cite{Lloyd2020quantum} provide another algorithm for implementing the polar decomposition of matrix $A$. Since $A$ is not necessarily a unitary, they embed $A$ in a Hermitian matrix $H_A$
\begin{equation}\label{eq:H}
H_A := \frac{1}{r} \left[\begin{array}{cc} 0 & A^{\dag} \\ A & 0\end{array}\right] = (A+A^{\dag})/r.
\end{equation}
and exponentiate it with Hamiltonian simulation, i.e. implementing $U_A = e^{i H_A t}$ as a result. $U_A$ is then used inside quantum phase estimation to decompose any input state into the eigenbasis of $H_A$. Finally, based on the estimates produced by phase estimation, the amplitude of each eigenvector in the decomposition is multiplied by the sign of its corresponding eigenvalue. This algorithm is described in greater detail in Section \ref{subsec:faster_procrustes} and Appendix \ref{appsubsec:DME_procrustes} when we consider its application to the Procrustes problem.

In the vanilla polar decomposition of \cite{Lloyd2020quantum} (assuming a readily-available $H_A$), $\eps$ cannot be chosen entirely independently of $\kappa$. \footnote{We thank Seth Lloyd for pointing this out to us.} In order to resolve the sign of the smallest eigenvalue of $H_A$, $\eps$ should be on the order of $\sigma_{\rm min}(A)$, or $1/\kappa$. However, in applications such as the Procrustes problem and $\PGM$, when $H_A$ is not readily available and its implementation itself incurs errors, $\eps$ should be $\min(1/\kappa,\eps_{\rm all})$ where $\eps_{\rm all}$ depends both on the final desired error and the errors incurred by other steps in the procedure. In the rest of this paper, we present algorithms for these applications which also speed up those of \cite{Lloyd2020quantum} by sidestepping the need for costly density matrix exponentiation (used to enact $U_A$). In place of $U_A$, we use the unitaries from block-encoding the natural input assumptions for the problem, as inputs to our building block Algorithm \ref{algo:polar}.

\section{Pretty-good measurements}\label{sec:PGM}
The {\em pretty-good measurement}, or $\PGM$, \cite{Belavkin75,Belavkin75a,Hausladen93bach,hausladen94}, is a close-to-optimal measurement for distinguishing ensembles of quantum states, and a prominent theoretical tool in quantum information and computation. Many quantum algorithms rely ultimately on reduction to a state discrimination problem. Because PGMs are close to optimal for state discrimination, they play a prominent role in achieving the algorithm's optimal performance. For example, $\PGM$s are the optimal measurement for the dihedral hidden subgroup problem \cite{bacon2005optimalMeasurementDiahedralHSP}; in quantum search with wildcards \cite{ambainis2014wildcards}; and also in quantum probably-approximately-correct learning \cite{arunachalam2017OptQSampCoplLearn}. 

$\PGM$s are also instrumental to the theory of quantum measurements. In a breakthrough paper \cite{haah2017OptTomography}, Haah \emph{et al.} proved that $O(D^2)$ copies of an unknown $D$-dimensional quantum state, are necessary and sufficient to estimate its density matrix to $\epsilon$ error. To prove sufficiency, they showed that a measurement inspired by the mathematical form of pretty-good measurements can achieve optimal sample complexity for this fundamental problem. Ref. \cite{jezek02} has also proposed a method for computing optimal POVMs for minimum-error detection of quantum states, that consists in iteratively performing the $\PGM$.

Ref. \cite{Lloyd2020quantum}, building on work by \cite{HolPGM78} and \cite{eldar2002optimalframes}, observed that the polar decomposition makes for an implementation of $\PGM$. Our $\PGM$ Algorithm \ref{algo:PGM} uses this core observation and a novel method of block-encoding the ensemble. Algorithm \ref{algo:PGM} also provides a cubic speedup vis-a-vis the general-purpose algorithm for {\em pretty-good instruments} due to Ref. \cite{gilyen2020quantum} -- which interestingly, also uses QSVT on an `ensemble state'. Our contribution is thus an alternative (and faster) use of the QSVT toolbox to realize this ubiquitous measurement to discriminate between pure states. We note that our methods can also be used to implement Holevo's pure state measurement \cite{HolPGM78,tyson10}.

\subsection{Pretty-Good Measurements}\label{subsec:PGM_def}
Suppose we are given a quantum system which is in one of the unknown quantum states $\alpha_j \in \mathcal{D}(\mathcal{H})$ where $j\in [r]$, with $r$ known.
The quantum system is in $\alpha_j$ with known probability $p_j$, which corresponds to the ensemble $\mathcal{E} = \{p_j, \alpha_j\}_{j\in[r]}$. We would like to perform a measurement, characterized by a set of Hermitian, positive semi-definite measurement matrices $\mathcal{M} = \{M_i\}_{i\in [r+1]}$, $\sum_{i=0}^{r+1} M_i = \id_{\mathcal{H}}$, to maximize the average probability of identifying which $\alpha_j$ corresponds to the system. The optimal measurement to distinguish states in the ensemble $\mathcal{E}$ is the one that attains the following expected success probability
\be\label{eq:successp}
P^{\rm opt}(\mathcal{E} ) = \max_{\mathcal{M}} \sum_i p_i \Tr(M_i \alpha_i).
\ee
It is not known in general how to write down the optimal measurement for ensembles of more than three states. However, the $\PGM$, also known as the {\em square-root measurement}, is a specific POVM that performs at most quadratically worse than the optimal measurement. Denoting its success probability as $P^{\PGM}$, we have
\be
P^{\rm opt}(\mathcal{E} )^2 \leq P^{\PGM}(\mathcal{E} ) \leq P^{\rm opt}(\mathcal{E} ).
\ee
In this section, we specialize to the case when the ensemble consists of $n$-qubit pure states.
\begin{defn}[$\PGM$ for pure states]\label{defn:PGM}
Given a pure-state ensemble $\mathcal{E} = \{p_i, \ket{\phi_i}\}_{i=0}^{r-1}$ where the $\ket{\phi_i}$ are $n$-qubit states in $\mathcal{H}$ and $\sum_i p_i= 1, p_i >0$, define the (unnormalized) vectors 
\be\label{eq:nu_def}
\ket{\nu_i} := (\sum_j p_j \ket{\phi_j}\bra{\phi_j})^{-1/2} \sqrt{p_i} \ket{\phi_i} = \rho^{-1/2}  \sqrt{p_i} \ket{\phi_i}, 
\ee 
where  $\rho = \sum_{i=0}^{r-1} p_i\ket{\phi_i}\bra{\phi_i}$, and the inverse is understood to be a pseudoinverse. The pretty-good measurement is the POVM $\{M_i\}_{i\in[r+1]}$ where
\begin{align}
M_i = \ket{\nu_i}\bra{\nu_i} \quad \forall i \in [r] \quad \text{and} \quad M_r = \mathbbm{1}_{\mathcal{H}} - \sum_{i=0}^{r-1}M_i, 
\end{align}
where the last operator corresponds to an `inconclusive' measurement outcome.
\end{defn}
Thus, with the assumption that input state $\ket{\omega}$ is drawn from the ensemble \footnote{For simplicity, though it is not necessary for our algorithm, we assume we are promised that $\omega$ is one of the states in the ensemble, as is the case when the $\PGM$ is used for state discrimination.}, an ideal implementation of the $\PGM$ on $\ket{\omega}$ would result in the following measurement probabilities (stored in the vector $q_{\PGM}$):
\begin{equation}\label{eq:PGM_ideal}
q_{\PGM}[i]:=\left\{\begin{array}{ll}
\left|\left\langle\nu_{i} \mid \omega\right\rangle\right|^{2} & 0\leq i \leq r-1 \\
0 & i = r \text{ (inconclusive)}
\end{array}\right.
\end{equation}
It can be checked that $\sum_i \ketbra{\nu_i}{\nu_i} = \Pi_{{\rm row} (\sum_{i\in[r]} \ketbra{\phi_i}{\phi_i})}$, and this ensures that $\sum_{i\in[r]} q_{\PGM}[i] = 1$. The vector $q_{\PGM}$ represents the {\em ideal} probabilities of a standard basis measurement on a quantum state prepared by our algorithm. In the actual implementation of the $\PGM$, a standard basis measurement on the algorithm's output state has a finite probability of landing on inconclusive outcomes $i\geq r$. However since we demand closeness of the actual output distribution to $q_{\PGM}$ in total variation distance, it suffices to group all inconclusive outcomes into a single, $r+1$-th outcome.

\subsubsection{Access assumptions}\label{subsubsec:access_PGM_polar}
We will assume access to the following unitaries:
\begin{enumerate}
    \item State preparation unitaries
    \begin{align}\label{eq:stateprepU}
    U_{\rm copy}\ket i |0\rangle  = \ket i \ket{i}; \qquad
   U_{\phi}\ket i |0\rangle  = \ket i \ket{\phi_i} \qquad \forall i\in [r]
   \end{align}
   as well as their controlled versions. Note that $U_{\rm copy}$ is a layer of CNOT gates, so we will not include it in the total gate complexity count. Let $T_{\phi}$ denote the gate complexity of $U_{\phi}$.
   \item Quantum sample access to $\vec{p}$ (using terminology of \cite{aharononv2007AdiabaticQStateGeneration}) which is the following unitary:
   \be\label{eq:tUp}
U_p\ket{\bar 0}= \sum_{k \in [r]} \sqrt{p_k}  \ket{k}.
\ee
Let $T_{p}$ denote the gate complexity of $U_p$ . 
\end{enumerate}
In case quantum sample access is not available, but $\vec{p}$ is known classically, we can prepare a Grover-Rudolph oracle using a quantum-random access memory storing a classical tree of partial sums at a cost of $\Ord{r {\rm poly} \log r}$, and $\Ord{{\rm poly} \log r}$ per oracle query \cite{Grover2002,giovannetti2008quantum}. 

\subsection{Pretty-Good Measurements from quantum polar decomposition}\label{subsec:faster_PGM}
We now explain how the polar decomposition enables one to implement the $\PGM$ associated with ensemble $\mathcal{E}=\left\{p_{i},\left|\phi_{i}\right\rangle\right\}_{i\in[r]}$ on the input state $\omega = \ketbra{\omega}{\omega}$. The following argument also generalizes that of \cite{Lloyd2020quantum} to nonuniform priors. Define matrix $A:=\sum_{j=0}^{r-1} \sqrt{p_j} \ket j \bra {\phi_j}$, where $\{\ket{j}\}_{j\in[r]}$ are standard basis states. To perform the desired $\PGM$, it suffices to implement  $U_{\rm polar}(A)$ on $\omega$, and then measure in the standard basis. For this choice of $A$, 
\begin{align}
    U_{\rm polar}(A) = A(A^{\dag}A)^{-1/2}  = \sum_{i=0}^{r-1} \sqrt{p_i} \ket{i}\bra{\phi_i} \Big( \sum_{k=0}^{r-1} p_k \ket{\phi_k}\bra{\phi_k} \Big)^{-1/2} = \sum_{i=0}^{r-1}  \ket{i}\bra{\nu_i},\label{eq:upolar_PGM}
\end{align}
and the transformation we have described takes $\omega$ to
\begin{align}\label{eq:PGM}
    U_{\rm polar}(A) \omega U_{\rm polar}(A)^{\dag} 
    = \sum_{i,j} \bra{\nu_i}\omega\ket{\nu_j} \ketbra{i}{j}. 
\end{align}
Then the probability of measuring $\ket{i}$, $i\in [r]$, on this state, is
\begin{align}\label{eq:PGMoutcomes}
\Pr(i)&=\text{Tr}(\ket{i}\bra{i} \sum_{i,j} \bra{\nu_i}\omega\ket{\nu_j} \ketbra{i}{j}) = \text{Tr}(\ket{\nu_i}\bra{\nu_i}\omega) = q_{\PGM,i}(\ket{\omega}).
\end{align}
 as desired (c.f. Eq.~\eqref{eq:PGM_ideal}). 
 
 Next, we show how to block-encode $A$ based on our access assumptions. This means that, in fact we will be working with an enlarged input state and unitaries. Nevertheless, if we now let $\Pr(i)$ be the probability of measuring $\ket{\bar 0}\ket{i}$, $i\in[r]$, our Eq.~\eqref{eq:PGMoutcomes} still holds, due to the assumption that $\omega$ is in the rowspace of $A$.\footnote{In case $\omega$ is not guaranteed to be entirely in the rowspace of $A$, an additional amplitude amplification step is required in our algorithm because $U_{\rm polar}(A)$ redistributes the part of the input state outside ${\rm row}(A)$ into the part of the enlarged Hilbert space outside the  $\ketbra{0^a}{0^a}$ subspace.}

Instead of block-encoding $A:=\sum_{j=0}^{r-1} \sqrt{p_j} \ket j \bra {\phi_j}$ directly, we more generally show how to block encode a weighted sum of $r$ projectors:
\begin{align}
&\sum_{k=0}^{r-1} \sqrt{p_k s_k} \ket{\psi_k}\bra{\phi_k} \quad \nonumber
\end{align}
where $\vec{p}, \vec{s}$ are probability vectors on $[r]$ and $\{\ket{\psi_k},\ket{\phi_k}\}_{k\in[r]}$ are $n$ qubit states.
Here we require quantum sample access to the probability distributions $p$ and $s$ (though other forms of access to $p, s$ will work too -- see remarks), and the state preparation unitaries $U_{\psi}$ and $U_{\phi}$. Our method was inspired by that of \cite{gilyen2018QSingValTransfArXiv} to block-encode a density matrix given a unitary preparing its purification.

\begin{lemma}[Block-encoding a weighted sum of projectors from quantum samples]\label{lem:block_qsamples}
A $(1, 0,\log_2(r)+n)$-block-encoding of $\sum_{k=0}^{r-1} \sqrt{p_k s_k} \ket{\psi_k}\bra{\phi_k}$ can be implemented with 1 use each of $U_{\rm \phi}^{\dag}$, $U_{\psi}$, $U_s^{\dag}$ and $U_p$; and a SWAP gate on $n$ qubits.
\end{lemma}
\begin{proof}
Define the following unitaries acting on $n+g = n +\log_2(r)$ qubits: 
\begin{align}
G_1 = U_{\psi}\cdot (U_p \otimes \id_n) \quad &i.e.\quad G_1 \ket{\bar{0}}_g \ket{\bar{0}}_n = \sum_k \sqrt{p_k} \ket{k}_g\ket{\psi_k}_n \label{eq:G1}\\
G_2 = U_{\phi} \cdot (U_s\otimes \id_{n})\quad  
&i.e. \quad G_2 \ket{\bar{0}}_g \ket{\bar{0}}_n = \sum_k \sqrt{s_k} \ket{k}_g\ket{\phi_k}_n\label{eq:G2}
\end{align}
We are interested in preparing a block-encoding of $\sum_{k=0}^{r-1} \sqrt{p_k s_k} \ket{\psi_k}\bra{\phi_k}$. The $(i,j)$-th element of this matrix is
\be\label{eq:Aij}
\bra{i} \sum_{k=0}^{r-1} \sqrt{p_k s_k} \ket{\psi_k}\bra{\phi_k} \ket{j} = \sum_{k\in [r]}\sqrt{p_k s_k} \braket{i}{\psi_k} \braket{\phi_k}{j}
\ee
We claim that the following unitary prepares the desired block-encoding:
\be \label{eq:epic_unitary}
(G_2^{\dag} \otimes \id_n) (\id_n \otimes SWAP_{n,n}) (G_1 \otimes \id_n).
\ee
To prove this, we evaluate the $(i,j)$-th element in the top-left-block of the unitary given above for $i,j\in 2^n$, and show that it is equal to Eq.~\eqref{eq:Aij}. 
\begin{align}
   & \bra{\bar{0}}_g \bra{\bar{0}}_n \bra{i}_n (G_2^{\dag} \otimes \id_n) (\id_g \otimes SWAP_{n,n}) (G_1 \otimes \id_n) \ket{\bar{0}}_g \ket{\bar{0}}_n \ket{j}_n\\
   & = \sum_{l'} \sqrt{s_{l'}}\bra{l'}_g \bra{\phi_{l'}}_n \bra{i}_n (\id_g \otimes SWAP_{n,n}) \sum_l \sqrt{p_l} \ket{l}_g \ket{\psi_l}_n\ket{j}_n\\
& = \sum_{l'}\sqrt{s_{l'}} \bra{l'} \bra{\phi_{l'}} \bra{i}  \sum_l \sqrt{p_l} \ket{l} \ket{j}\ket{\psi_l}\\
   & = \sum_{l,l'} \sqrt{s_{l'}p_l} \braket{l'}{l} \braket{\phi_{l'}}{j} \braket{i}{\psi_l} = \sum_{l} \sqrt{s_{l}p_l} \braket{\phi_{l}}{j} \braket{i}{\psi_l}
\end{align}
\end{proof}


Note that, alternatively, if instead of quantum samples we only have access to controlled single rotations acting as 
\be\label{eq:rotation_Up}
\tUp\ket{k}\ket{0} = \ket{k} (\sqrt{p_k} \ket{0} + \sqrt{1-p_k} \ket{1} )\qquad \forall k\in [r].
\ee
we can also obtain an $(r, 0,\log_2(r)+1)$ block-encoding of $A$ directly, with only $g+2$ ancilla qubits but with a quadratically worse subnormalization factor compared to Lemma \ref{lem:block_qsamples}. We refer the reader to the remarks below Lemma \ref{lem:block} in Section \ref{sec:procrustes} for the method.

Finally, our algorithm to implement the $\PGM$ is as follows.
\begin{algorithm}[H]
		\textbf{Input:} Unitaries $U_{\phi}$ and $U_p$, input state $\ket{\omega}$, accuracy parameter $\eps$. \\
		\textbf{Output:} A measurement outcome distributed approximately as $q_{\PGM}$.
		\vspace{5pt}
\begin{algorithmic}[1]
\State Prepare $U_{A/\sqrt{r}}$, a $(\sqrt{r}, 0, \log_2(r)+n)$-block encoding of $A = \sum_{k=0}^{r-1} \sqrt{p_k} \ket{k}\bra{\phi_k}$ using Lemma \ref{lem:block_qsamples}.
\State On input state $\ket{\omega}$, apply Algorithm \ref{algo:polar} with $U_A \gets U_{A/\sqrt{r}}, \eps \gets \eps^2$. Call the resulting state $\ket{\omega_{\rm out}}$.
\State Measure $\ket{\omega_{\rm out}}$ in the standard basis. 
\end{algorithmic}
\caption{Pretty-good measurement on pure states}
\label{algo:PGM}
\end{algorithm}
Note that, in the application of Algorithm \ref{algo:polar}, there is no need for amplification if $\ket{\omega}$ is drawn from the ensemble. Theorem \ref{thm:PGM} gives the guarantees of Algorithm \ref{algo:PGM}.
\begin{theorem}[\label{thm:PGM}Pretty-Good Measurement associated with
$\mathcal{E} = \{ (p_i,\ket{\phi_i})_{i \in {[r]}} \}$ ]
Let $A:= \sum_{j=0}^{r-1} \sqrt{p_j} \ket j \bra {\phi_j}$, and let $\kappa_A:= 1/{\sigma_{\min}(A)}$. Given access to $U_{\phi}$, and  $U_{p}$ and some input state $\ket{\omega}$ drawn from $\mathcal{E}$, the distribution $\tilde{q}$ over measurement outcomes output by Algorithm \ref{algo:PGM} satisfies $d_{TV}(\tilde{q},q_{\PGM})\leq \eps$, with $\Ord{\kappa_A \sqrt{r} \log(1/\varepsilon)}$ uses each of $U_{p}$ and $U_{\phi}^{\dag}$.
\end{theorem}

\begin{proof}

By the assumption that $\ket{\omega}\in \mathcal{E}$, and choosing the accuracy parameter $\eps^2$ in Theorem \ref{thm:polar}, we obtain the guarantee
\be\label{eq:statecloseness_repeated}
 \left\lVert \ketbra{\bar{0}}{\bar{0}}\ket{\omega_{\rm out}}  -  \ket{\bar{0}}\ket{\rm ideal} \right\rVert_2 \leq \eps^2
\ee
where $\ket{\rm ideal} = U_{\rm polar}(A) \ket{\omega}$. Define $\ket{\omega_{\rm ideal}} := \ket{\bar{0}}\ket{\rm ideal}$, then $\lVert \ket{\omega_{\rm out}} - \ket{\omega_{\rm ideal}} \rVert_2^2 \leq O(\eps^2)$ by using a similar argument to the proof of Corollary \ref{corr:polar}.
Since $\lVert \ket{\omega_{\rm ideal}} - \ket{\omega_{\rm out}} \rVert_2^2 = 2(1- \text{Re}(\braket{\omega_{\rm ideal}}{\omega_{\rm out}}))$, this implies that  $\text{Re}(\braket{\omega_{\rm ideal}}{\omega_{\rm out}}))^2 \geq 1- O(\eps^2)$. Also, trace distance on pure states is given by $\frac{1}{2} \lVert  \ketbra{\omega_{\rm ideal}}{\omega_{\rm ideal}} - \ketbra{\omega_{\rm out}}{\omega_{\rm out}} \rVert_1 = \sqrt{1-|\braket{\omega_{\rm ideal}}{\omega_{\rm out}}|^2}$ and we thus obtain that the trace distance between $\ket{\omega_{\rm ideal}},\ket{\omega_{\rm out}}$ is bounded by:
\be \lVert  \ketbra{\omega_{\rm ideal}}{\omega_{\rm ideal}} - \ketbra{\omega_{\rm out}}{\omega_{\rm out}} \rVert_1 \leq O(\eps). \ee 
Now, as defined in the Theorem, let $\tilde{q}$, $q_{\PGM}$ be the probability distributions over standard basis measurements on $\ket{\omega_{\rm out}}$ and $\ket{\omega_{\rm ideal}}$, respectively.  Using the fact that $d_{TV}(\tilde{q},q_{\PGM}) \leq  \frac{1}{2} \lVert  \ketbra{\omega_{\rm ideal}}{\omega_{\rm ideal}} - \ketbra{\omega_{\rm out}}{\omega_{\rm out}} \rVert_1$, we therefore obtain the desired bound on output probabilities: 
$d_{TV}(\tilde{q},q_{\PGM}) \leq O(\eps).$ 

Now we prove the claimed gate complexity. From Lemma \ref{lem:block_qsamples}, $U_{A/\sqrt{r}}$ can be implemented with 1 use each of $U_{\rm copy}$, $U_p$ and $U_{\phi}^{\dag}$, $\log_2(r)$ Hadamards; and a SWAP gate on $n$ qubits. Also note that the effective condition number of $U_{A/\sqrt{r}}$ is $\sqrt{r}\kappa_A$ due to its subnormalization factor of $\sqrt{r}$. Finally, Theorem \ref{thm:polar} gives the required number of uses of $U_{A/\sqrt{r}}$ and its inverse.
\end{proof}

\subsection{Comparison to the Petz-based algorithm}\label{subsec:PGM_comparison}
To our knowledge, there are two other proposals for implementing pretty-good measurements: the ``density matrix exponentiation (DME)-based algorithm" of Ref. \cite{Lloyd2020quantum} which is also based on the quantum polar decomposition, and the ``Petz-based algorithm" of Ref.~\cite{gilyen2020quantum}, which is based on specializing their algorithm for Petz recovery channels. We show that for the pure-state $\PGM$ problem, Algorithm \ref{algo:PGM} is faster than these alternatives -- showing a cubic speedup over the Petz-based algorithm, and an exponential speedup over the DME-based algorithm. In this section, we focus on analyzing the former; the speedup over the DME-based algorithm is inherited from that for the quantum Procrustes problem, which we discuss in the next section.
 
\subsubsection{Ideal Petz transformation}\label{subsubsec:ideal_Petz}

Given the ensemble $\mathcal{E}= \{p_j,\sigma_B^j\}_{j\in[r]}$, define the mixed state $\overline{\sigma}_{B}:=\sum_{j}p_j\sigma_{B}^{j}$. Ref.~\cite{gilyen2020quantum} show that the following transformation is a special case of the Petz map: 
\begin{align}\label{eq:PGI}
    \omega_B \longrightarrow
\sum_{j=0}^{r-1} \ket{j}\bra{j}_{X}\otimes p_j\left(  \sigma_{B}^{j}\right)
^{\frac{1}{2}}
\left(\overline{\sigma}_{B}\right)^{-\frac{1}{2}}
\omega_B
\left(\overline{\sigma}_{B}\right)^{-\frac{1}{2}}
\left(
\sigma_{B}^{j}\right)  ^{\frac{1}{2}}.
\end{align}
The above transformation amounts to a pretty-good {\em instrument}. An instrument is a generalization of a measurement where both a classical ($X$) and quantum ($B$) register are prepared, the former of which is to be measured in the computational basis. In the case of pure state ensembles, $\sigma_B^j = \ket{\phi_j}\bra{\phi_j}$, then  $\overline{\sigma}_{B} = \sum_{j=1}^r p_j \ket{\phi_j}\bra{\phi_j}$ and the above may be written more simply as
\begin{align}
    \omega & \longrightarrow
\sum_{j=0}^{r-1} \ket{j}\bra{j}\otimes p_j \left( \ket{\phi_j}\bra{\phi_j}\right)
^{\frac{1}{2}}
\left( \sum_{l=0}^{r-1} p_l \ket{\phi_l}\bra{\phi_l}\right)^{-\frac{1}{2}}
\omega
\left( \sum_{k=0}^{r-1} p_k\ket{\phi_k}\bra{\phi_k}\right)^{-\frac{1}{2}}
\left(
\ket{\phi_j}\bra{\phi_j}\right)  ^{\frac{1}{2}}\\
& = \sum_{j=0}^{r-1} p_j \ket{j}\bra{j}\otimes  \ket{\phi_j}\bra{\phi_j}
\rho^{-\frac{1}{2}}
\omega
\rho^{-\frac{1}{2}}
\ket{\phi_j}\bra{\phi_j}=\sum_{j=0}^{r-1} \bra{\nu_j} \omega \ket{\nu_j} \ketbra{j}{j}_X\otimes  \ketbra{\phi_j}{\phi_j}_B. 
\label{eq:PGMpetz}
\end{align}
Measuring system $X$ in the standard basis then yields outcomes distributed as $q_{\PGM}$. The steps to implement a $\PGM$ are now manifest: use Ref. \cite{gilyen2020quantum}'s algorithm for enacting the Petz map to implement the transformation of Eq.\eqref{eq:PGMpetz}, then measure system $X$ in the standard basis. 

\subsubsection{Comparison}
The following Theorem sums up the accuracy and complexity guarantees of the Petz-based algorithm when used on pure state ensembles. Noting that this algorithm prepares a different quantum state from our Algorithm \ref{algo:PGM}, to compare the two algorithms we simply demand that the final measurement's output distribution is $\eps$-close to the ideal $q_{\PGM}$.
\begin{theorem}[\label{thm:PGI} Pretty-Good Instrument associated with
$\mathcal{E} = \{ (p_i,\ket{\phi_i})_{i \in {[r]}} \}$ ]
Let $A:= \sum_{j=0}^{r-1} \sqrt{p_j} \ket j \bra {\phi_j}$, and let $\kappa_A:= 1/{\sigma_{\min}(A)}$. Given access to $U_{\phi}$, and  $U_{p}$ and some input state $\ket{\omega}$ drawn from $\mathcal{E}$, the state in \eq{eq:PGMpetz} can be prepared with a gate complexity of
\begin{align}\label{eq:final_performance-petz-map}%
\tOrd{
\kappa_A^3\sqrt{r} T_p + (\kappa_A^3\sqrt{r} + r^{3/2}\kappa_A)T_{\phi}} \quad \text{(Petz)}.
\end{align}
and the distribution $\tilde{q}$ over measurement outcomes on the state satisfies $d_{TV}(\tilde{q},q_{\PGM})\leq \eps$.
\end{theorem}

\begin{proof}
Firstly, we put the error criteria on the same footing. Theorem 1 of \cite{gilyen2020quantum} provides an algorithm that achieves the transformation of Eq.~\eqref{eq:PGMpetz} up to a diamond-norm error of $\eps$. Given input state $\omega$, let $\mathcal{P}$ be the ideal Petz map with the choices of $\mathcal{N}, \sigma$ that implement the PGM, $\chi$ be the ideal output state and $q_{\PGM}$ be the probability distribution over standard basis measurements of the $X$ register of $\chi$. Let $\tilde{\mathcal{P}}, \tilde{\chi}, \tilde{q}$ represent the actual versions of these objects implemented by the Petz algorithm. Indeed, as Section \ref{subsubsec:ideal_Petz} shows, $\tilde{q}_{i} := \Tr(\left(\ketbra{i}{i}_X\otimes \mathbbm{1}_B\right) \tilde{\chi}) = q_{\PGM, i}$. We now argue that $\left\Vert \mathcal{\tilde{P}}-\mathcal{P} \right\Vert_{\diamond}\leq\eps$ implies that $\lVert \tilde{q} - q_{\PGM} \rVert_1 \leq \eps.$ Now for any POVM $\{M_i\}$,
\begin{align}
\left\lVert \mathcal{\tilde{P}}-\mathcal{P} \right\rVert_{\diamond}\leq\eps 
\rightarrow \lVert \chi - \tilde{\chi} \rVert_1 \leq \eps \quad \forall \omega \rightarrow \sum_i \left|\Tr{(M_i\chi)} -\Tr{(M_i\tilde{\chi})}\right| \leq \eps 
\end{align} 
where the first implication is a property of the diamond norm and the second implication is the relation between trace distance on states and TV-distance of their respective measurement probabilities (see for instance Exercise 9.1.10 in \cite{wilde2017QIT}). Now choosing $M_i = \ketbra{i}{i}_X\otimes \id_B$ we obtain the desired TV-distance guarantee.

Secondly, we compute the complexity of the Petz-based algorithm's required inputs, in the pure-state $\PGM$ setting. The Petz-based algorithm requires:
\begin{enumerate}
\item A block-encoding of the state $\sigma_{XB}= \sum_{j}p_j\ketbra{j}{j}_{X}\otimes\ketbra{\phi_j}{\phi_j}_{B}$. Call its gate complexity $T_{\sigma}$ and define 
\begin{align}\label{eq:K_s}
\kappa_{\sigma} &:= 1/\sigma_{\rm min}(\sigma) = 1/\sigma_{\rm min}( \sum_{j}p_j\ketbra{j}{j}_{X}\otimes\ketbra{\phi_j}{\phi_j}_{B}) =1/\min_{j} p_{j}\geq r.
\end{align}
In terms of our access assumptions, we claim that 
\be\label{eq:T_s}
T_{\sigma}= 2T_{p} + 2 T_{\phi} + T_{\rm SWAP}.
\ee
This can be seen from Lemma 45 of \cite{gilyen2018QSingValTransfArXiv} and noting that the unitary $U_{\phi} \cdot U_{\rm copy} \cdot U_p$ on input state $\ket{\bar 0}\ket{\bar 0}\ket{\bar 0}$ prepares a purification of the state $\sigma_{XB}$.

\item A block-encoding of the unitary extension of the channel $\mathcal{N}_{XB\rightarrow B}= \operatorname{Tr}_{X}$. Call its gate complexity $T_{\mathcal{N}}$. This is zero in our case, as the extension of the trace-out channel is simply the unitary that acts with the identity, which requires no gates to implement.

\item A block-encoding of the state $\mathcal{N}_{XB\rightarrow B}(\sigma_{XB}) = \overline{\sigma}_{B} = \sum_{j\in[r]} p_j \ket{\phi_j}\bra{\phi_j}.$ Call its gate complexity $T_{\mathcal{N}(\sigma)}$ and define 
\begin{align}\label{eq:K_Ns}
\kappa_{\mathcal{N}(\sigma)} := 1/\sigma_{\rm min}(\overline{\sigma}_B)= 1/\sigma_{\rm min}(\sum_{j=1}^r p_j \ket{\phi_j}\bra{\phi_j}) =1/\sigma_{\rm min}(A)^2 = \kappa_A^2.
\end{align}
In terms of our access assumptions, noting that the unitary $U_{\phi}  \cdot U_p$ on input state $\ket{\bar 0}\ket{\bar 0}$ prepares a purification of the state $\mathcal{N}(\sigma_{XB})$ and employing similar reasoning as before, we have
\be\label{eq:T_Ns}
T_{\mathcal{N}(\sigma)} = 2T_{p} + T_{\rm SWAP}.
\ee
\end{enumerate}

Thirdly and finally, we compute the gate complexity. The algorithm of \cite{gilyen2020quantum} prepares the desired state with a gate complexity of 
\begin{equation}\label{eq:performance-petz-map}%
\tOrd{\sqrt{d_{E}\kappa_{\mathcal{N}(\sigma)}}
\left(
\kappa_{\mathcal{N}(\sigma)} T_{\mathcal{N}(\sigma)}
\!+\! T_{\mathcal{N}} \!+\! T_{\sigma}\min\left(\kappa_{\sigma}, d_E \kappa_{\mathcal{N}(\sigma)}/{\epsilon^2}\right)\right)},
\end{equation}
where $d_E$ is the dimension of the environment system they append, which is not smaller than the Kraus rank of $\mathcal{N}(\cdot)$. In this case $d_E \geq r$. 
Plugging Eq.s \eqref{eq:K_s}--\eqref{eq:T_Ns} into \eq{eq:performance-petz-map}, we see that the first option of the minimization should be taken and we obtain the claimed complexity:
\begin{align}\label{eq:final_performance-petz-map-repeat}%
\tOrd{
\kappa_A^3\sqrt{r} T_p + (\kappa_A^3\sqrt{r} + r^{3/2}\kappa_A)T_{\phi}} \quad \text{(Petz)}.
\end{align}
\end{proof}
On the other hand, we see from Theorem \ref{thm:PGM} that our polar-based Algorithm \ref{algo:PGM}'s runtime, in terms of the number of uses of $U_p$ and $U_{\phi}$, is
\begin{equation}
\tOrd{\sqrt{r} \kappa_A(T_{p}+T_{\phi})} \quad \text{(Polar)}.
\end{equation}
which is a cubic improvement -- in terms of $\kappa_A$ dependence of uses of $U_p$; and in terms of $\max(\kappa_A,\sqrt{r})$ dependence of uses of $U_{\phi}$ -- over the Petz-based algorithm (\eq{eq:final_performance-petz-map}) for pure-state $\PGM$.


\section{Quantum Procrustes problem}\label{sec:procrustes}
The polar decomposition finds many uses in low-rank approximations and optimization problems because it possesses a {\em best approximation} property. This is instantiated by the Procrustes problem of learning optimal unitary transforms: finding an orthogonal matrix $U$ that most nearly transforms a given matrix $F$ to a given matrix $G$, where the error criterion is the sum-of-squares of the residual matrix $UF-G$. Here, $F, G$ may represent data matrices with errors; or they may represent matrices whose columns are vectors that one wishes to orthogonalize. In the quantum setting, the latter case has a particularly appealing interpretation as finding an optimal unitary transform between pairs of input and output quantum states. 

We provide an algorithm for this problem with runtime $\tOrd{{\rm poly}(r,\kappa,\log(1/\eps))}$. We also rigorously reformulate an alternative algorithm of \cite{Lloyd2020quantum} and prove that its runtime is $\tOrd{{\rm poly}(r,\kappa,\eps)}$, with larger dependencies on $r,\kappa$. Since pretty-good measurements are an instance of the Procrustes transformation, this also completes our argument that our $\PGM$ algorithm, when applied to uniform ensembles, speeds up that of \cite{Lloyd2020quantum}.

\subsection{Procrustes transformation and input assumptions}\label{subsec:procrustes}
Suppose we are given unitaries preparing $r$ input/output pairs of $n$-qubit states, $\{(\ket{\psi_i}, \ket{\phi_i}) : i\in [r]\}$, not guaranteed to be orthogonal. The hypothesis is that there is one common isometry that takes each input state to its corresponding output state. (WLOG, we assume the input and output states have the same dimensions because we can always pad the smaller-dimensional states with zero ancillae.) Though not required, the interesting case is when $r$ is much smaller than $N$. It may be verified that the polar decomposition isometry is an optimal transformation between inputs and outputs, that is, it is a minimizer for the following ``least-squares" error criterion:
\begin{equation}\label{eq:minimization}
U^\ast := \arg \min_{U\,{\rm isometry}} \lVert UF -G \rVert_2^2, \quad \text{(Procrustes isometry)}
\end{equation}
where $F \in \mathds{C}^{N\times r}$ is the matrix whose columns are the $\left|\psi_{j}\right\rangle,$ and $G\in \mathds{C}^{N \times r}$ is the matrix whose columns are the $\left|\phi_{j}\right\rangle .$ More precisely, the polar isometry associated with $A:=GF^{\dag} = \sum_{i=0}^{r-1} \ket{\phi_i} \bra{\psi_i}$ is a nonunique solution for $U^{\ast}$ -- and we will thus call it $U_{\rm Procrustes}$. Our polar decomposition algorithm will thus allow us to implement $U_{\rm Procrustes}$ directly -- without the need for doing state tomography on the input-output pairs and classically computing the minimizer $U^{\ast}$. 

Let $g=\log_2(r)$. We will assume we have quantum circuits that implement state preparation unitaries, $U_{\psi}$, $U_{\phi}$ acting on $g+n$ qubits as
\begin{align}\label{eq:uncontrolledU}
    U_{\psi}\ket i |0\rangle  = \ket i \ket{\psi_i}; \qquad
   U_{\phi}\ket i |0\rangle  = \ket i \ket{\phi_i}. 
\end{align}
By replacing all gates in the circuit with their controlled versions, we can also implement controlled state preparation unitaries which we call $cU_{\psi}$ and $cU_{\phi}$. 

Let us contextualize the assumption of state preparation unitaries. Often, the interesting applications of the quantum Procrustes transformation are when the output states are a known orthogonal basis set, and the input states are non-orthogonal states that one would like to find the closest orthogonal approximation of. Indeed, the $\PGM$ of a uniform state ensemble $\{1/r,\ket{\phi_i}\}_{i\in[r]}$, is simply the Procrustes transformation with input states $\{\ket{\phi_i}\}$ and output states $\{\ket{i}\}$. Similarly, the method of L\"{o}wdin orthogonalization from quantum chemistry takes a set of input wavefunctions with desirable properties, and outputs an orthogonal basis of wavefunctions. 
\subsection{Quantum Procrustes with QSVT}\label{subsec:faster_procrustes}
We show how to transform the state preparation unitaries into a block-encoding of $A = GF^{\dag} = \sum_{i\in [r]}\ket{\phi_i}\bra{\psi_i}$. We end up with a sub-normalization factor of $r$, which, with no guarantees on the state pairs, is necessary to ensure that the operator norm of the block containing $A$ is at most $1$ (for example, consider the case where the state pairs are close-to-identical.) 

\begin{lemma}[Block-encoding from state preparation unitaries]\label{lem:block}
An $(r, 0,\log(r)+1 )$-block-encoding of $A = GF^{\dag} = \sum_{i=0}^{r-1}\ket{\phi_i}\bra{\psi_i}$ can be implemented with 2 uses of $cU_{\phi}$ and $cU_{\psi}$ and $\Ord{n+\log(r)}$ additional quantum gates. 
\end{lemma}

\begin{proof}
We shall use a reflection unitary $W$, defined as
\begin{equation}
    W := \mathbbm{1}_{n+g} - 2  \mathbbm{1}_{g} \otimes \ket{\bar{0}}\bra{\bar 0}_n,
\end{equation}
which can be prepared using $\Ord{n+g}$ gates where $g=\log(r)$.
Let $H^{\otimes g}$ denote $g$ parallel applications of the Hadamard gate $H$. Define the unitaries
\begin{align} 
\label{eq:unitaries}
    V_0 := (H^{\otimes g} \otimes \mathbbm{1}_{n}) U_\phi U_\psi^\dagger (H^{\otimes g} \otimes \mathbbm{1}_{n}) \quad \text{and} \quad
    V_1 := (H^{\otimes g} \otimes \mathbbm{1}_{n}) U_{\phi} W U_{\psi}^{\dag} (H^{\otimes g} \otimes \mathbbm{1}_{n}).
\end{align}
 We consider the following linear combination of the unitaries $V_0$ and $V_1$:
\begin{align}
\frac{1}{2}(V_0-V_1) &= (H^{\otimes g} \otimes \mathbbm{1}_{n}) U_{\phi} \left(\sum_{k=0}^{r-1}\ket{k}\bra{k}_g \otimes  \ket{\bar{0}}\bra{\bar 0}_n\right) U_{\psi}^{\dag} (H^{\otimes g} \otimes \mathbbm{1}_{n}). \label{eq:1}\\
&= H^{\otimes g} \otimes \mathbbm{1}_n \left(\sum_{k=0}^{r-1}  \ket{k}\bra{k}_g \otimes \ket{\phi_k} \bra{\psi_k}_n \right) H^{\otimes g} \otimes \mathbbm{1}_n \\
&= \frac{1}{r}
    \sum_{k=0}^{r-1} \big(\sum_{k', k''} (-1)^{k\cdot (k'+k'')} \ket {k'} \bra {k''}_g\big) \otimes \ket{\phi_k} \bra{\psi_k}_n  \\
    &=\frac{1}{r} 
    \sum_{k=0}^{r-1}  \ket{\bar{0}}\bra{\bar 0}_g \otimes \ket{\phi_k}\bra{\psi_k}_n +  \frac{1}{r} \sum_{k', k'' \neq 0,0}  (-1)^{k\cdot (k'+k'')} \ket {k'} \bra {k''}_g \otimes \ket{\phi_k} \bra{\psi_k}_n,
\end{align}
where in the second equality we have used the well-known identity
$H^{\otimes g} \ket k = \frac{1}{\sqrt r}\sum_{k'=0}^{r-1} (-1)^{k\cdot k'} \ket {k'}$.
That is to say, $\frac{1}{2}(V_0 - V_1)$ in the computational basis is
\begin{equation}\label{eq:M}
    \frac{1}{2}(V_0 - V_1) = \left[\begin{array}{cc} \frac{1}{r} \sum_{k=0}^{r-1}\ket{\phi_k}\bra{\psi_k} & \cdot \\ \cdot & \cdot\end{array}\right] = \left[\begin{array}{cc} A/r & \cdot \\ \cdot & \cdot\end{array}\right],
\end{equation}
which has a block-structure. $\frac{1}{2}(V_0 - V_1)$ is not a unitary because the operator $\sum_{k=0}^{r-1}\ket{k}\bra{k} \otimes  \ket{\bar{0}}\bra{\bar 0} = \mathbbm{1}\otimes \ket{\bar{0}}\bra{\bar 0}$ is not a unitary. However, it is straightforward to block-encode this operator using the Linear Combination of Unitaries (LCU) method of Berry et al.~\cite{berry2015HamSimNearlyOpt}: define a controlled unitary
\be \label{eq:U}
    cV_0V_1 &:= \ket 0\bra 0_{c} \otimes V_0 + \ket 1\bra 1_c \otimes V_1.
\ee
This unitary can be implemented efficiently since we can implement controlled versions of $U_{\phi}$ and $U_{\psi}$ and $W$ efficiently. Finally, using the definitions of the $X,H$ gates, it is easily verified that the following circuit is a $(r,0,g+1)$-block-encoding of $A$:
\begin{equation}\label{eq:A_block}
(X H \otimes \mathbbm{1}_{n+g}) (cV_0V_1) (H\otimes \mathbbm{1}_{n+g})=
\frac{1}{2} \left[\begin{array}{cc} V_0-V_1 & V_0+V_1 \\ V_0+V_1 & V_0-V_1\end{array}\right] = \left[\begin{array}{cc} A/r & \cdot \\ \cdot & \cdot\end{array}\right],
\end{equation}
where the last equality follows from Eq.~\eqref{eq:M}. Aside from the two uses of $cU_{\psi}$ and $cU_{\phi}$, the additional gate complexity is $\Ord{n+g}$ from applying the controlled versions of $W$ and $H$.
\end{proof}
We make two remarks: firstly, the method of Lemma \ref{lem:block_qsamples} with $\vec{p},\vec{q}\leftarrow {\rm Unif}_r$, and $U_p, U_q\leftarrow H^{\otimes g}$ alternatively yields a $(r, 0, \log_2(r)+n)$-block-encoding of $A$, but this has more ancilla qubits. Secondly, the block-encoding method of Lemma \ref{lem:block} can also block-encode the $\PGM$ matrix of Section \ref{sec:PGM} directly from access to the unitary $\tUp$ of \eq{eq:rotation_Up}, with
\begin{align}
    \tilde{V}_0 &:= (H^{\otimes g} \otimes \mathbbm{1}_{n}) U_{\rm copy}\otimes \mathbbm{1} \tUp^\dagger U_{\phi}^\dagger (H^{\otimes g} \otimes \mathbbm{1}_{n})\\ \qquad
    \tilde{V}_1 &:= (H^{\otimes g} \otimes \mathbbm{1}_{n}) U_{\rm copy}  \left(\mathbbm{1}_{n+g+1} - 2  \mathbbm{1}_{g} \otimes \ket{\bar{0}}\bra{\bar 0}_{n+1}\right)\tUp^\dagger U_{\phi}^\dagger (H^{\otimes g} \otimes \mathbbm{1}_{n})
\end{align}
and the final block-encoding unitary $(X H \otimes \mathbbm{1}_{n+g+1}) (c\tilde{V}_0\tilde{V}_1) (H\otimes \mathbbm{1}_{n+g+1}) $.

We are now ready to state and prove our algorithm for enacting the Procrustes transformation on an input state $\ket{\psi}$. In the following theorems, letting $\Pi_{{\rm row}(F^{\dag})}$ be the projector onto ${\rm Span} (\{\ket{\psi_i}\}_{i\in[r]})$, the ideal output state of the Procustes transformation is
\begin{equation}
\ket{\psi_{\rm ideal}} := \frac{1}{\lVert \Pi_{{\rm row}(F^{\dag})} \ket{\psi} \rVert_2} U_{\rm Procrustes} \ket{\psi}.
\end{equation}
\begin{theorem}[\label{thm:QSVT_procrustes}QSVT-based approach for Procrustes problem] 
Given an $n$-qubit input state $\ket{\psi}$ such that $\lVert \Pi_{{\rm row}(F^{\dag})} \ket{\psi} \rVert_2 \geq c$, a state $\ket{\psi_{\rm out}}$ satisfying
\be\label{eq:statecloseness_procrustes}
 \left\lVert \ket{\psi_{\rm out}}  - \ket{\bar{0}}\ket{\psi_{\rm ideal}} \right\rVert \leq \varepsilon,
 \ee
can be prepared with $\tOrd{\frac{\kappa r}{c} \log(1/\varepsilon)}$ uses each of $cU_{\psi}, cU_{\phi}$ and their adjoints; and $\Ord{n+\log(r)}$ additional quantum gates per use, where $\kappa:= 1/{\sigma_{\min}(A)}$.
\end{theorem}
\begin{proof}
Our algorithm is simply to use Lemma \ref{lem:block} to prepare the unitary $U_{A/r}$, and use that as the unitary input to Algorithm \ref{algo:polar}. The effective condition number of $U_{A/r}$ is $\kappa r$. Theorem \ref{thm:polar}.2 shows that Algorithm \ref{algo:polar} requires $\tOrd{\frac{\kappa r}{c} \log(1/\varepsilon)}$ applications of $U_{A/r}$ and its adjoint, each of which requires 2 applications of $cU_{\psi}$ and $cU_{\phi}$ and their adjoints (Lemma \ref{lem:block}).
\end{proof}

For comparison, we now rigorously present an alternative algorithm for Procrustes based on density matrix exponentiation (``the DME-based method"), originally sketched out in \cite{Lloyd2020quantum}. This algorithm uses filter functions and quantum phase estimation, similar to the HHL algorithm for solving linear systems of equations \cite{HHL}. As in \cite{HHL}, we shall use the following `helper' state
\begin{equation}
\left|\Psi_{0}\right\rangle :=\sqrt{\frac{2}{T}} \sum_{\tau=0}^{T-1} \sin \frac{\pi\left(\tau+\frac{1}{2}\right)}{T}|\tau\rangle
\end{equation}
which we will prepare in a control register and tensor into $\ket{\psi}$ for the phase estimation to act on. 
\begin{algorithm}[H]
		\textbf{Input:} State preparation unitaries $cU_{\phi}$, $cU_{\psi}$, input state $\ket{{\rm in}}$ such that $\ket{{\rm in}}\ket{0} = \sum_j \beta_j \ket{u_j}$, parameters $t_0 \in [0,\infty), T \in \mathbb{Z}_+$.\\
		\textbf{Output:} An approximation to $U_{\rm procrustes}\ket{\rm in}$ in the sense of Theorem \ref{thm:QSVT_DME}.
		\vspace{5pt}
\begin{algorithmic}[1]
\State Prepare $\ket{\Psi_0}_c$ up to error $\varepsilon_{\psi}$.  
\State Apply the Quantum Phase Estimation unitary $\tilde{P}$ as follows: apply the conditional Hamiltonian evolution $U_{\rho} = \sum_{\tau=0}^{T-1}|\tau\rangle\langle\tau|_c \otimes e^{i (H + \id )t_{0}\tau / 2T}$ on $\ket{\Psi_0}_c \ket{{\rm in}}\ket{0}$, then apply the inverse Quantum Fourier Transform to the control register. This produces the state
\begin{equation}\label{eq:state_afterQPE}
\sum_{j\in[N]} \sum_{k\in[T]} \alpha_{k \mid j} \beta_{j}|k\rangle_Q \left|u_{j}\right\rangle 
\end{equation}
where for a given pair $(j,k)$, each $k$ corresponds to a possible estimate of the eigenvalue $\lambda_j$, with amplitude $|\alpha_{k \mid j}|^2$. We shall call the corresponding estimate $\tilde{\lambda}_k := 2(2\pi k/t_0 - 1/2)$. 
\State Tensor in a flag register $\ket{0}_{\rm flag}$ and apply the controlled rotation $U_{\rm flag}$ to it, conditioned on the $\ket{k}_Q$ register. $U_{\rm flag}$ acts as follows:
\be\label{eq:hdef}
\ket{k}_Q\ket{0}_{\rm flag} \rightarrow \ket{k}_Q \left( \sqrt{1-f(\tilde{\lambda}_k)^2} \ket{\rm ill}_{\rm flag} + f(\tilde{\lambda}_k) \ket{\rm well}_{\rm flag}\right)
\ee
Here, $f:[-1,1]\rightarrow [-1,1]$ is a filter function meant to approximate the sign function.
\begin{equation}\label{eq:filter}
f(x) := 
\begin{cases}
      -1 & x< -\frac{1}{\kappa r} \\
\sin\left(\frac{\pi}{2} \frac{x+\frac{1}{2\kappa r}}{\frac{1}{\kappa r}-\frac{1}{2\kappa r}} \right) & -\frac{1}{\kappa r}  \leq x < -\frac{1}{2\kappa r}  \\
0 & -\frac{1}{2\kappa r} \leq x < \frac{1}{2\kappa r}\\
\sin\left(\frac{\pi}{2} \frac{x-\frac{1}{2\kappa r}}{\frac{1}{\kappa r}-\frac{1}{2\kappa r}} \right) & \frac{1}{2\kappa r} \leq x < \frac{1}{\kappa r}\\
      1 & x\geq \frac{1}{\kappa r} \\
    \end{cases} 
\end{equation}
where $\kappa = \frac{1}{\sigma_{\rm min}(\sum_{i \in [r]} \ketbra{\phi_i}{\psi_i})}$ (in practice a lower bound suffices).
\State Uncompute the $\ket{k}_Q$ register by applying $\tilde{P}^{\dag}$.
\State (Optional) Post-select on the $\ket{\rm well}$ outcome by performing amplitude amplification. 
\end{algorithmic}
\caption{DME-based algorithm for the Procrustes problem, originally sketched in \cite{Lloyd2020quantum}}
\label{algo:DME}
\end{algorithm}
\begin{theorem}[DME-based approach for Procrustes problem]
\label{thm:QSVT_DME}
Given an $n$-qubit input state $\ket{\psi}$ such  that $\lVert \Pi_{{\rm row}(F^{\dag})} \ket{\psi} \rVert_2 > c$, a quantum state $\ket{\psi_{\rm out}}$ satisfying
\be\label{eq:statecloseness_procrustes_appendix}
 \left\lVert \ket{\psi_{\rm out}}  - \ket{\bar{0}}\ket{\psi_{\rm ideal}} \right\rVert_2 \leq \varepsilon,
 \ee
can be prepared with $\tOrd{\frac{r^2 \kappa^2}{c^{3/2}\eps^3}}$ uses each of $cU_{\psi}, cU_{\phi}$ and their adjoints. 
\end{theorem}
The proof is provided in Appendix \ref{sec:DME}. Compared to HHL, our proof involves additional steps to handle the substantial $1/\eps$ gate complexity scaling of density matrix exponentiation -- the building block of the phase estimation. Comparing Theorems \ref{thm:QSVT_procrustes} and \ref{thm:QSVT_DME}, the QSVT-based method for enacting the Procrustes transformation leads to exponential speedups in $\eps$ and polynomial speedups in $r$ and $\kappa$ compared to the DME-based method (Algorithm \ref{algo:DME}).

\bibliography{refs}
\bibliographystyle{unsrt}

\appendix
\section{Proofs for Quantum Polar Decomposition (Section \ref{subsec:faster_polar})}\label{app:proofs_fasterpolar}

We first state the following theorem from prior work, which is a standard method to amplify desired parts of a quantum state `obliviously' -- i.e. without needing to use a unitary that prepares the initial state $\ket{\psi}$. In our work, this theorem will be used to amplify the part of a given initial state that is in the rowspace of $A$.
   
\begin{theorem}[Robust oblivious amplitude amplification (Theorem 28 of \cite{gilyen2018QSingValTransfArXiv})\label{thm:OAA}] Let $n \in \mathbb{N}_{+}$ be odd, let $\varepsilon \in \mathbb{R}_{+}$, let $U$ be a unitary, let $\widetilde{\Pi}, \Pi$ be orthogonal projectors, and let $W: \operatorname{img}(\Pi) \mapsto \operatorname{img}(\widetilde{\Pi})$ be an isometry, such that
\begin{equation}\label{eq:cond_OAA}
|| \sin \left(\frac{\pi}{2 n}\right) W|\psi\rangle-\widetilde{\Pi} U|\psi\rangle || \leq \varepsilon
\end{equation}
for all $|\psi\rangle \in \operatorname{img}(\Pi) .$ Then we can construct a unitary $\tilde{U}$ such that for all $|\psi\rangle \in \operatorname{img}(\Pi)$
$$
|| W|\psi\rangle-\widetilde{\Pi} \tilde{U}|\psi\rangle || \leq 2 n \varepsilon
$$
which uses a single ancilla qubit, with $n$ uses of $U$ and $U^{\dagger}, n$ uses of $C_{\Pi}$ NOT and $n$ uses of $C_{\widetilde{\Pi}}$ NOT gates and $n$ single qubit gates.
\end{theorem}

We first prove a helper Lemma about Theorem \ref{thm:SVT} which leads to a simple Corollary about applying the polar decomposition isometry. 
\begin{lemma}\label{lem:thm1_cons}
From Theorem 1 with $U\gets U_A$, $\epsilon \in (0,1)$, $\delta \leq \sigma_{\min}(A)$, we obtain
\begin{equation}\label{eq:provethis_1}
    \lVert \tilde{\Pi} U_{\Phi} \Pi - WV^{\dag} \rVert \leq \eps.
\end{equation}
\end{lemma}
\begin{proof}
We remind readers that $W, V$ refer to the SVD factors of matrix $\ket{0^a}\bra{0^a} \otimes A$. We first observe that some blocks in $U_{\Phi}$ in Theorem \ref{thm:SVT} are identically $0$. That is,
\begin{equation}\label{eq:0}
\tilde\Pi \Pic \tilde\Pi  \ U_{\Phi}\ \Pi  \Pin \Pi = \tilde \Pi \Picn \tilde\Pi \ U_{\Phi}\ \Pi  \Pir \Pi = \tilde\Pi \Picn \tilde\Pi \ U_{\Phi}\ \Pi  \Pin \Pi = 0.
\end{equation}
To prove \eq{eq:0}, we show that $\tilde\Pi \Pic \tilde\Pi  \ U_{\Phi}\ \Pi  \Pin \Pi = 0$; a similar calculation can be performed for all the other equalities. 
From Ref.~\cite{gilyen2018QSingValTransfArXiv}, and recalling the definition of $P^{(SV)}$ in Section \ref{subsec:notation}, Theorem 1 prepares a $U_{\Phi}$ such that
\be
\tPi U_{\Phi} \Pi = P^{(SV)}(\tilde \Pi U_A \Pi),
\ee
where $P$ is an odd polynomial (Lemma 25 of Ref.~\cite{gilyen2018QSingValTransfArXiv}). Hence,
 \be\label{eq:PSV}
 \tilde \Pi \Pic \tilde \Pi  \ U_{\Phi}\ \Pi  \Pin \Pi =  \tilde \Pi \Pic P^{(SV)}(\tilde\Pi U_A \Pi)\Pin \Pi.
 \ee
With the SVD of $A$ we have $\tilde \Pi U_A \Pi = \ketbraz \otimes \sum_j \sigma_j \ket{{w_j}} \bra{{v_j}}$.
Since $P$ is an odd polynomial, $P(0) = 0$. 
Then, we have
 \be 
 \tilde\Pi \Pic \tilde{\Pi}  \ U_{\Phi}\ \Pi  \Pin \Pi &=& \sum
_{j} P(\sigma_j) \tilde\Pi \Pic \ketbraz\otimes \ket{{w_j}} \bra{{v_j}} \Pin \Pi. \label{eq:picpin}
 \ee
Since $\bra{v_j} \Pin =0$ for all $ \ket{v_j}$,
the right-hand-side of Eq. \eqref{eq:picpin} is zero as desired.

The choice of $\delta$ implies that in Theorem \ref{thm:SVT}, $\tilde{\Pi}_{\geq \delta}=\tilde \Pi \Pic \tilde\Pi$ and $\Pi_{\geq \delta}=\Pi \Pir \Pi$ and its guarantee becomes $\left\lVert \tilde \Pi \Pic \tilde\Pi  U_{\Phi} \Pi \Pir \Pi -\tilde \Pi \Pic \tilde\Pi \left(W V^{\dagger}\right) \Pi \Pir \Pi \right\rVert \leq \varepsilon$.
 Observe that the second term simplifies to $WV^{\dag}$ and from \eq{eq:0} the first term is
\begin{align}
\tilde\Pi \Pic \tilde \Pi \ U_{\Phi}\ \Pi \Pir \Pi = \tilde \Pi (\Pic + \Picn) \tilde\Pi \ U_{\Phi}\  \Pi (\Pir + \Pin)\Pi = \tilde \Pi U_{\Phi} \Pi ,
\end{align}
Hence, we obtain Eq.~\ref{eq:provethis_1}.
\end{proof}

\begin{proof}[Proof of Corollary \ref{corr:polar}]
\eq{eqCor1} follows immediately from Lemma \ref{lem:thm1_cons} and Definition \ref{defn:polar}. Then for $\ket{\psi_{\rm ideal}} := \ket{0^a} U_{\rm polar} \ket{\psi}$,
\begin{align}
\left \lVert \ket{\psi_{\rm out}} -\ket{\psi_{\rm ideal}} \right \rVert^2 &= \left \lVert \tPi \ket{\psi_{\rm out}} -\ket{\psi_{\rm ideal}} \right \rVert^2 + \left \lVert (\id-\tPi) \ket{\psi_{\rm out}} \right\rVert^2\leq \eps^4 + 1 - (1-\eps^2)^2 = \Ord{\eps^2}. \label{eq:states_3}
\end{align}
where in the inequality we have used $\lVert \tPi \ket{\psi_{\rm out}} \rVert_2 \geq 1-\eps^2$. This follows from \eq{eqCor1}, the triangle inequality and the fact that $U_{\rm polar}$ is an isometry, so $\Vert\ket{\psi_{\rm ideal}}\Vert_2 = 1$ if $\ket{\psi}\in {\rm row}(A)$. Taking the square root of \eq{eq:states_3} yields \eq{eqCor1_2}. 
\end{proof} 

\begin{proof}[Proof of Theorem \ref{thm:polar}]
There exists $n \in \mathbb{N}_+$ odd such that $\vert \sin(\frac{\pi}{2n})- c\vert < \eps^2/4n$. Note that $n=\Ord{1/c}$, where $c\equiv \lVert \Pir \ket{0^a}\ket{\psi} \rVert_2 = \lVert WV^{\dag} \ket{0^a}\ket{\psi}\rVert_2$. In step (2), we apply Theorem \ref{thm:SVT} with the parameters $U\gets U_A$, $\delta \leq \sigma_{\min}(A)$, $\tPi = \Pi \gets \ket{0^a}\bra{0^a}\otimes \id$ and $\eps \gets \frac{\eps^2}{4n}$. In step (3), we apply Theorem \ref{thm:OAA} with the parameters $U\gets U_{\Phi}, \tPi = \Pi \gets \ket{0^a}\bra{0^a}\otimes \id$, $\ket{\psi} \gets \ket{0^a}\ket{\psi}$.

We will first prove that these choices yield the guarantee on $\ket{\psi_{\rm out}}$
\be \label{eq:states_1_Thm3}
\left\lVert \ketbra{0^a}{0^a}\ket{\psi_{\rm out}} - \frac{1}{c} WV^{\dag}\ket{0^a} \ket{\psi} \right\rVert \leq \eps^2.
\ee
We now show that \eq{eq:states_1_Thm3} follows from Theorem \ref{thm:OAA}.
Let $X$ be some isometry that maps $\rm{img}(\Pi) \rightarrow \rm{img}(\tPi) $, that in particular maps 
\be
\ket{0^a}\ket{\psi} \overset{X}{\longrightarrow} \frac{1}{c} WV^{\dag}\ket{0^a}\ket{\psi}.
\ee
Note that the right-hand-side is normalized by our definition of $c$. Then the condition \eq{eq:cond_OAA} is satisfied, with
\begin{align}
\left \lVert \sin \left(\frac{\pi}{2 n}\right) X \ket{0^a}\ket{\psi}-\widetilde{\Pi} U_{\Phi} \ket{0^a}\ket{\psi}\right \rVert &\leq \left \lVert c X \ket{0^a}\ket{\psi}-\widetilde{\Pi} U_{\Phi}\ket{0^a} \ket{\psi} \right \rVert + \frac{\eps^2}{4n}\\
&= \left \lVert  WV^{\dag}\ket{0^a} \ket{\psi} -\widetilde{\Pi} U_{\Phi} \ket{0^a}\ket{\psi} \right \rVert + \frac{\eps^2}{4n}\leq \frac{\eps^2}{2n}
\end{align}
where the second inequality follows from Corollary \ref{corr:polar}. Theorem \ref{thm:OAA} thus gives us 
\begin{align}\label{eq:states_2_Thm3}
\left\lVert X \ket{0^a}\ket{\psi} - \tPi U_{AA} \ket{0^a}\ket{\psi} \right\rVert &= \left \lVert  \frac{1}{c}WV^{\dag}\ket{0^a} \ket{\psi} -\widetilde{\Pi} \ket{\psi_{\rm out}} \right \rVert \leq \eps^2,
\end{align}
which is \eq{eq:states_1_Thm3}. Using a similar argument to Corollary \ref{corr:polar}, this implies the state closeness guarantee \eq{eq:statecloseness}. The claimed gate complexity follows from observing that Oblivious Amplitude Amplification requires $n=\Ord{1/c}$ uses of $U_{\Phi}$, $U_{\Phi}^\dag$, each of which requires $m = \Ord{\kappa \log(\frac{4n}{\eps^2})}$ uses of $U_A, U_A^{\dag}$, $C_{\ketbra{0^a}{0^a}}NOT$ and single-qubit gates.
\end{proof}

\section{Quantum Procrustes through density matrix exponentiation (Theorem \ref{thm:QSVT_DME}) \label{sec:DME}}

In this section, we refine the density matrix exponentiation (DME)-based approach to quantum Procrustes in \cite{Lloyd2020quantum}, and we also rigorously prove its gate complexity. 

On a high level, the transformation in Theorem \ref{thm:QSVT_DME} is accomplished by applying the polar transform of the matrix $A = GF^{\dag}$ to the good part of the input state (call it $\ket{{\rm in}}$), and then (optionally) amplifying the result. However, because $A$ is not necessarily a unitary \footnote{We assume $A \in \mathbb{C}^{N\times N}$ is square, but all our arguments easily generalize to non-square $A$s.}, we embed it in some Hermitian matrix $H\in \mathbb{C}^{2N \times 2N}$ 

\begin{equation}\label{eq:H}
H := \frac{1}{r} \left[\begin{array}{cc} 0 & A^{\dag} \\ A & 0\end{array}\right] = (A+A^{\dag})/r.
\end{equation}
Exponentiating $H$ results in some unitary $U$; $U$ is used inside a quantum phase estimation procedure on $\ket{{\rm in}}$, and then an appropriately-chosen function is applied to the eigenvectors to accomplish the polar transform. 

To set the groundwork for subsequent sections, we now explain how the singular vectors and values of $A$ are related to those of the embedding matrix $H$. Let the singular value decomposition of $A$ be $A = \sum_{i\in[r]} \sigma_i \ket{w_i} \bra{v_i}$, where we may also choose an arbitrary spanning basis of the nullspace of $A$ and $A^{\dag}$ denoted by $\{\ket{v_i^{\perp}}\}_{i\in [N-r]}$, $\{\ket{w_i^{\perp}}\}_{i\in [N-r]}$ respectively. We will denote the rowspace of $A$ as `well', for `well-conditioned subspace'; and we will denote $A$'s nullspace as `ill', for `ill-conditioned' subspace. Then the eigenvalues and eigenvectors of $H$ are

\begin{align}
\left\{\lambda_j, \ket{u_j} \right\}_{j\in [2N]} &= \left\{ \left(\frac{\sigma_i}{r}, \ket{+_i} \right)\right\}_{i\in [r]} \bigcup \left\{ \left(-\frac{\sigma_i}{r}, \ket{-_i} \right) \right\}_{i\in [r]} \quad \text{(rowspace)}\\
&\bigcup \left\{ \left(0, \left[\begin{array}{c} \ket{v_i^{\perp}} \\ 0 \end{array}\right] \right)\right\}_{i\in [N-r]} \bigcup \left\{ \left(0, \left[\begin{array}{c} 0 \\ \ket{w_i^{\perp}} \end{array}\right] \right)\right\}_{i\in [N-r]} \quad \text{(nullspace)}
\end{align}
where we have defined $\ket{+_i}:=\left[\begin{array}{c} \ket{v_i}\\ \ket{w_i} \end{array}\right]$ and $\ket{-_i}:=\left[\begin{array}{c} \ket{v_i}\\ -\ket{w_i} \end{array}\right]$. 

This section is organized as follows:

\begin{itemize}
    \item In Section \ref{appsubsec:DME_procrustes}, we give a detailed and rigorous description of the DME-based algorithm which is our Algorithm \ref{algo:DME}.
    \item In Section \ref{appsubsec:boundingerror}, we analyze the error of this algorithm. We consider two sources of error: the inherent error due to quantum phase estimation (QPE) analyzed in Sections \ref{appsubsubsec:error_QPE}, and the error due to enacting the unitary in QPE by density matrix exponentiation analyzed in Sections \ref{appsubsubsec:error_LMR} and \ref{appsubsubsec:error_DME}.
    Finally, in Section \ref{appsubsubsec:proof_DME}, we put these components together to prove Theorem \ref{thm:QSVT_DME}.
\end{itemize}

\subsection{Detailed description of DME-based algorithm for polar decomposition}\label{appsubsec:DME_procrustes}

In this subsection we rigorously present the DME-based algorithm to apply the polar decomposition. We start off with a simplified description of the algorithm. Given any input vector $\ket{{\rm in}} \in \mathbb{C}^N$, we may write it in terms of the singular vectors of $A$ $\{\ket{v_i}\}_i$ as well as the nullspace basis $\{\ket{v_i^{\perp}}\}_i$. Then the Procrustes transformation in Theorem \ref{thm:QSVT_DME} is first, to perform
\be \label{eq:firsttransform}
\ket{{\rm in}} =\sum_{i\in \text{well}} c_i \ket{v_i} + \sum_{i\in \text{ill}} d_i \ket{v_i^{\perp}} \longrightarrow \sum_{i\in \text{well}} c_i \ket{w_i} + \sum_{i\in \text{ill}} d_i \ket{v_i^{\perp}}
\ee
and secondly, to postselect on the part of the state in `well'. To accomplish the transformation in Eq.~\eqref{eq:firsttransform}, we proceed in the following steps:
\begin{enumerate}
\item Enlarge $\ket{\rm in}$ by adjoining a $\ket0$ ancilla, forming the state
\be
\ket{\rm in}\ket{0} = \sum_{i \in \ro(A)} c_i \ket{v_i}\ket{0} + \sum_{i \in {\rm null}(A)} d_i \ket{v_i^{\perp}}\ket{0} = \sum_{i \in \ro(A)} \frac{c_i}{2} (\ket{+_i}+\ket{-_i}) + \sum_{i \in {\rm null}(A)} d_i \ket{v_i^{\perp}}\ket{0}.
\ee 
\item Use quantum phase estimation\cite{Cleve_1998,Kitaev1997}(QPE) with the unitary $e^{iH}$ to store an estimate of each eigenvector's associated eigenvalue in an ancilla register (labelled with `Q'). The unitary $e^{iH}$ is enacting using density matrix exponentiation of the state preparation unitaries $cU_{\psi}, cU_{\phi}$. Observe that $H = \rho-\tilde{\rho}$ where
\[\rho = \frac{1}{2r} \left[\begin{array}{cc} E & A^{\dag} \\ A & \tilde{E} \end{array}\right] \quad \text{and} \quad \tilde{\rho}= \frac{1}{2r} \left[\begin{array}{cc}E & -A^{\dag} \\ -A & \tilde{E} \end{array}\right],\] and $E:= \sum_{j}\left|\psi_{j}\right\rangle\left\langle\psi_{j}\right|$ and $\tilde{E}:= \sum_{j}\left|\phi_{j}\right\rangle\left\langle\phi_{j}\right|$. Therefore, exponentiating $H$ amounts to exponentiating $\rho$, $\tilde{\rho}$, each of which can be 
prepared with a single use of $cU_{\psi}$ and $cU_{\psi}$\footnote{That is, a single use of $cU_{\psi}$ and $cU_{\psi}$ prepares the state 
$\frac{1}{\sqrt{2 r}} \sum_{j=1}^{r}|j\rangle \otimes\left(\begin{array}{l}\left|\psi_{j}\right\rangle \\ \left|\phi_{j}\right\rangle\end{array}\right)$ and the state in the second register is $\rho$. $\tilde{\rho}$ is prepared in a similar fashion, by tracing out the first register in the state $\frac{1}{\sqrt{2 r}} \sum_{j=1}^{r}|j\rangle \otimes\left(\hspace{-5pt}\begin{array}{r}\left|\psi_{j}\right\rangle \\ -\left|\phi_{j}\right\rangle\end{array}\right)$}.

The QPE has the effect of decomposing $\ket{{\rm in}}\ket{0}$ into the eigenbasis of $H$, resulting in the state
\be\label{eq:secondtransform}
\sum_{i \in \ro(A)} \frac{c_i}{2} \left(\ket{+_i} \ket{\frac{\sigma_i}{r}}_Q + \ket{-_i}\ket{-\frac{\sigma_i}{r}}_Q \right) + \sum_{i \in {\rm null}(A)} d_i \ket{v_i^{\perp}}\ket{0}\ket{0}_Q.
\ee
Here, we have labelled the `Q' register with the value of the estimate it stores.
\item Conditioning on the `Q' register, we may flip the sign of the amplitudes of the eigenvectors corresponding to negative eigenvalues, resulting in the state
\be\label{eq:thirdtransform}
\sum_{i \in \ro(A)} \frac{c_i}{2} \left(\ket{+_i} \ket{\frac{\sigma_i}{r}}_Q - \ket{-_i}\ket{-\frac{\sigma_i}{r}}_Q\right) + \sum_{i \in {\rm null}(A)} d_i \ket{v_i^{\perp}}\ket{0}\ket{0}_Q.
\ee
\item Invert the QPE, resulting in the state (we now omit the `Q' register, since it has been set back to $\ket{0}$ by the inversion)
\be
\sum_{i \in \ro(A)} \frac{c_i}{2} \left(\ket{+_i}  - \ket{-_i}\right) + \sum_{i \in {\rm null}(A)} d_i \ket{v_i^{\perp}}\ket{0} = \sum_{i \in \ro(A)} c_i \ket{w_i}\ket{1} + \sum_{i \in {\rm null}(A)} d_i \ket{v_i^{\perp}}\ket{0}.
\ee
\item Post-select on the part of the state labeled with $\ket{1}$. 
\end{enumerate}

A detailed version of the procedure just described is provided in Algorithm \ref{algo:DME} in the main text. Algorithm \ref{algo:DME} differs from the idealized one in three ways: firstly, for each eigenvalue, QPE in Step (2) does not output a single eigenvalue estimate $\ket{\frac{\sigma_i}{r}}_Q$, but rather, a superposition over estimates $\sum_k \alpha_{k|i} \ket{k}_Q$ where each $k$ represents some estimate $\tilde{\lambda}_k$ for $\frac{\sigma_i}{r}$; secondly, as in \cite{HHL}, we have approximated the sign function in Step (3) with a filter function $f:[-1,1]\rightarrow[-1,1]$ for numerical stability; and lastly, instead of post-selection onto the $\ket{1}$ subspace, the filter function we have crafted means that it suffices to post-select onto the $\ket{\rm well}$ subspace. 

To conclude this section, we state some preliminaries which will simplify the discussion of the error bound. 

Overall, in steps (2)-(4) of Algorithm \ref{algo:DME} we have applied the unitary
\begin{equation}\label{eq:unitary}
U_{\rm all} := \tilde{P}(\tilde{D})^\dagger U_{\rm flag} \tilde{P}(\tilde{D})
\end{equation}
where the $\tilde{D}$ stands for density matrix exponentiation and we use the notation $\tilde{P}(\tilde{D})$ to mean ``inexact QPE where the DME step is also inexact". 

The following Lemma about continuity of the filter function $f$ in Eq.~\eqref{eq:filter}, says that the Lipschitz constants of the map $\lambda \rightarrow \ket{h(\lambda)}$ and the filter function $f$ are both $O(\kappa r)$.
\begin{lemma}[Lipschitz Continuity\label{lem:lipschitz}]
Let $\lambda, \tilde{\lambda} \in [-1,1]$. Then
\begin{enumerate}
    \item The map $\lambda \rightarrow \ket{h(\lambda)}$ defined in Eq.~\eqref{eq:hdef} satisfies
     \be \label{eq:lipschitz_1}
\lVert\ket{h(\lambda)}- \ket{h(\tilde{\lambda})} \rVert \leq C \kappa r |\lambda- \tilde{\lambda}|.
\ee
and this implies $\text{Re}(\braket{h(\tilde{\lambda}_k)}{h(\lambda_j)})  \geq 1-C^2 \kappa^2 r^2 |\lambda- \tilde{\lambda}|^2 .$
    \item The filter function $f:[-1,1]\rightarrow [-1,1]$ defined in \eqref{eq:filter} satisfies
    \be \label{eq:lipschitz_2}
|f(\lambda)- f(\tilde{\lambda})| \leq C \kappa r |\lambda- \tilde{\lambda}|.
\ee
\end{enumerate}
where $C = \pi$.
\end{lemma}

\begin{proof}
\begin{enumerate}
\item The given map is continuous over the domain $[-1,1]$, and hence it suffices to bound the norm of $\frac{d}{d\lambda} \ket{h(\lambda)}$ over all the regions, in which case the Lipschitz constant is the maximum over all its local Lipschitz constants. When $|\lambda| > \frac{1}{\kappa r}$, $\ket{h(\lambda)} = \ket{\text{well}}$ and hence $\frac{d}{d\lambda} \ket{h(\lambda)} = 0.$ The same is true for the region $\frac{-1}{2\kappa r} <\lambda< \frac{1}{2\kappa r}$. When $\frac{1}{2 \kappa r}< \lambda <\frac{1}{ \kappa r}$, \be
\ket{h(\lambda)} = \cos \left(\frac{\pi}{2} \frac{\lambda-1/\kappa r}{1/\kappa r-1/2\kappa r}\right) \ket{\text{ill}} + \sin \left(\frac{\pi}{2} \frac{\lambda-1/\kappa r}{1/\kappa r-1/2\kappa r}\right) \ket{\text{well}}
\ee
and 
\be
\frac{d}{d\lambda} \ket{h(\lambda)} = \kappa r \pi \left[ -\sin \left(\frac{\pi}{2} \frac{\lambda-1/\kappa r}{1/\kappa r-1/2\kappa r}\right) \ket{\text{ill}} + \cos \left(\frac{\pi}{2} \frac{\lambda-1/\kappa r}{1/\kappa r-1/2\kappa r}\right) \ket{\text{well}} \right]
\ee
and hence $\lVert \frac{d}{d\lambda} \ket{h(\lambda)} \rVert = \kappa r \pi$. A similar analysis yields the same bound on the derivative for the region $\frac{-1}{\kappa r}< \lambda< \frac{-1}{2 \kappa r}$.
\item Reasoning similar to above as the function is piecewise continuous over the domain $[-1,1]$. Regions $-1< \lambda< \frac{-1}{\kappa r}, \frac{-1}{2\kappa r} <\lambda< \frac{1}{2\kappa r}, \lambda>\frac{1}{\kappa r}$ have local Lipschitz constant $0$. Regions $\frac{-1}{\kappa r}< \lambda< \frac{-1}{2 \kappa r}$ and $\frac{1}{2 \kappa r}< \lambda<\frac{1}{ \kappa r}$ have local Lipschitz constant $\kappa \pi r$, as can be seen by taking the derivative, e.g. for the region $\frac{1}{2 \kappa r}< \lambda <\frac{1}{ \kappa r}$
\be
\frac{d}{d\lambda} \sin \left(\frac{\pi}{2} \frac{\lambda-1/2\kappa r}{1/\kappa r-1/2\kappa r}\right) = \pi \kappa r \cos \left(\frac{\pi}{2} \frac{\lambda-1/\kappa r}{1/\kappa r-1/2\kappa r}\right) 
\ee
and the absolute value of this derivative is bounded by $\pi \kappa r$ on the relevant region.
\end{enumerate}
\end{proof}
This is our analog of Lemmas 2 and 3 in HHL \cite{HHL} respectively. We now comment briefly on the differences. Because we only have two flag states (while they use three flag states), their $g$ function does not appear in our Lemma. It suffices for us to have two flag states, because the fact that we approximate the sign function (and not the inverse function) means that the flag $\ket{\text{ill}}$ gets zero amplitude in the well-conditioned subspace, and only has non-zero amplitude in the ill-conditioned subspace. 

\subsection{Bounding the error}\label{appsubsec:boundingerror}

We now analyze the error and gate complexity of this algorithm. We will consider two sources of error: 
\begin{enumerate}
\item In Section \ref{appsubsubsec:error_QPE}, we consider the inherent error of quantum phase estimation, which arises from the fact that quantum phase estimation on each eigenvector outputs a probability distribution over candidate eigenvalues, instead of putting all probability mass on the single correct candidate. 
\item In Section \ref{appsubsubsec:error_DME}, we consider the total error in performing density matrix exponentiation (which propagates into an error in quantum phase estimation).
\end{enumerate}
There is actually another source of error $\varepsilon_{\psi}$ from inexact preparation of the state $\ket \Psi_0$, but it can be made small with little effort \cite{HHL} and so we will ignore it. 

\subsubsection{Inherent error of Quantum Phase Estimation}\label{appsubsubsec:error_QPE}

The goal of this section is to bound the inherent error of quantum phase estimation, assuming an exact density matrix exponentiation subroutine (denoted by $D$, without a $\tilde{}$ ) within the phase estimation. That is, we show that the following two versions of Algorithm \ref{algo:DME} are not ``too different":
\begin{enumerate}
\item A version where the phase estimation suffers from inherent error. That is, we apply the unitary $\tilde{V} = \tilde{P}(D)^{\dag} U_{\rm flag} \tilde{P}(D)$. After post-selection on the flag register, call the resulting state $\ket{\tilde{x} }$.
\item  A version where the phase estimation does not suffer from inherent error. This corresponds to the unitary $V = P(D)^{\dag} U_{\rm flag} P(D)$. After post-selection  on the flag register, call the resulting (ideal) state $\ket{x}$.
\end{enumerate}
\begin{theorem}[Inherent error of quantum phase estimation\label{thm:QPE_DME}]
Assuming the ability to perform exact density matrix exponentiation, Algorithm \ref{algo:DME} applied to an input state $\ket{{\rm in}}$, satisfies the following:
\begin{enumerate}
    \item In the case when no post-selection is performed, the error is bounded as 
    \be
    \lVert \tilde{V} - V \rVert \leq O(r\kappa/t_0)
    \ee
    \item In the case when we post-select on the flag register being in the state $\ket{\text{well}}$, if the input state $\ket{\rm in}$ is such that $\lVert \Pi_{{\rm row}(F^{\dag})}\ket{\rm in} \rVert_2 \geq c > 0$,
    the error is bounded as 
    \be
    \lVert \ket{\tilde{x} } - \ket{x} \rVert \leq O\left(\frac{r \kappa}{\sqrt{c}t_0} \right).
    \ee
\end{enumerate}

\end{theorem}

\begin{proof}
Let $(\lambda_j, \ket{u_j})_{j}$ be the eigenvalues and eigenvectors of the Hamiltonian $H$ used in the quantum phase estimation. Then we write $\ket{{\rm in}}\ket{0} = \sum_j \beta_j \ket{u_j}$. 

We first prove Theorem \ref{thm:QPE_DME}.1. Consider two states resulting from applying $V$, $\tilde{V}$ respectively to the initial state (here $\ket{\rm others}$ represents all necessary ancillas):
\begin{align}
\ket{\varphi} = V \ket{\rm in}\ket{0}\ket{\rm others} &:= \sum_{j=0}^{N-1} \beta_{j} \ket{u_j} \ket{h(\lambda_j)} \label{eq:states_1}\\
\ket{\tilde{\varphi}} = \tilde{V} \ket{\rm in}\ket{0}\ket{\rm others} &:= \tilde{P}^{\dag} \sum_{j=0}^{N-1} \beta_{j}  \left|u_{j}\right\rangle
 \sum_{k=0}^{T-1} \alpha_{k|j} \ket{k} \ket{h(\tilde{\lambda}_k)}, \label{eq:states_2}
\end{align}
where $\alpha_{k|j}$ is the amplitude assigned to the basis vector $\ket{k}$ after QPE has been run for time $T$ to estimate the eigenvalue $\lambda_j$. In order to upper-bound $\lVert V - \tilde{V} \rVert$, it suffices to bound $\lVert \ket{\varphi} - \ket{\tilde{\varphi}} \rVert_2$ for any input state $\ket{{\rm in}}$. Now because $\lVert \ket{\varphi} - \ket{\tilde{\varphi}} \rVert_2 = \sqrt{2(1-\text{Re} \braket{\varphi}{\tilde{\varphi}} )}$, we may equivalently lower-bound 
\be
\text{Re}\braket{\varphi}{\tilde{\varphi}} = \text{Re}(\bra{\varphi} \tilde{P}^{\dag} \tilde{P} \ket{\tilde{\varphi}}) = \sum_{j=0}^{N-1} |\beta_j|^2 \sum_k |\alpha_{k|j}|^2 \text{Re}(\braket{h(\tilde{\lambda}_k)}{h(\lambda_j)}).
\ee
From Lemma \ref{lem:lipschitz}, however, we see that $\text{Re}(\braket{h(\tilde{\lambda}_k)}{h(\lambda_j)})  \geq 1-\pi^2 \kappa^2 (\lambda_j-\tilde{\lambda}_k)^2 r^2$ and, using the bound $|\alpha_{k|j}|^2 \leq O\left( \frac{1}{((\lambda_j - \tilde{\lambda}_k)t_0)^4}\right)$ and a similar calculation to the proof of Theorem 1 of \cite{HHL}, we conclude that $\text{Re}(\braket{\varphi}{\tilde{\varphi}}) \geq 1 - O(\kappa^2 r^2/t_0^2)$. The corresponding $\ell_2$-norm guarantee is Theorem \ref{thm:QPE_DME}.1.

We now prove Theorem \ref{thm:QPE_DME}.2. Here our slightly weaker Lipschitz guarantee in Lemma \ref{lem:lipschitz}.2 introduces some subtleties compared to HHL\cite{HHL}. After post-selection of states \eqref{eq:states_1}, \eqref{eq:states_2} onto the well-conditioned subspace via amplitude amplification, the states are 
\begin{align}
\ket{x} &:= \frac{\sum_{j=0}^{N-1} \beta_{j} \ket{u_j} f(\lambda_j)\ket{{\rm well}}}{\sqrt{\sum_j |\beta_j|^2 f(\lambda_j)^2}}. \label{eq:states_afterpost_1}\\
\ket{\tilde{x}} &:= \frac{\tilde{P}^{\dag} \sum_{j=0}^{N-1} \beta_{j}  \left|u_{j}\right\rangle
 \sum_{k=0}^{T-1} \alpha_{k|j} \ket{k} f(\tilde{\lambda}_k) \ket{{\rm well}}}{\sqrt{\sum_{k,j} |\alpha_{k|j}|^2 |\beta_j|^2 f(\tilde{\lambda}_k)^2}}. \label{eq:states_afterpost_2}
\end{align}
To simplify the above expressions, we define 
\begin{align}
\mathcal{W} &:= \{j\in [N-1]: |\lambda_j| > \frac{1}{\kappa r}\} \qquad (\mathcal{W} \text{ for ``well-conditioned")}\\
p &:= \sum_j |\beta_j|^2 f(\lambda_j)^2 = \sum_{j\in \W} |\beta_j|^2\label{eq:p}\\
\tilde{p} &:=  \sum_{k,j} |\alpha_{k|j}|^2 |\beta_j|^2 f(\tilde{\lambda}_k)^2,
\end{align} 
where the second equality in Eq.~\eqref{eq:p} follows from the fact that $f(\lambda_j)=0$ when $j \in \mathcal{W}$ and $|f(\lambda_j)|=1$ otherwise.

As before, it suffices to lower-bound the fidelity $F := \braket{x}{\tilde{x}} = \bra{x} \tilde{P}^{\dag} \tilde{P} \ket{\tilde{x}}$. Noting that $\tilde{P}\ket{x} = \frac{1}{\sqrt{p}}  \sum_{j \in \W}\beta_{j} \ket{u_j}  \sum_{k=0}^{T-1} \alpha_{k|j} \ket{k} f(\tilde{\lambda}_k) \ket{\text{well}} $, we shall proceed to rewrite $F$ following \cite{HHL}. To aid us, we define the probability distribution $\mathbb{P}_{K|J}:[T-1] \rightarrow [0,1]$ where $\mathbb{P}_{K|J} (k) := |\alpha_{k|j}|^2$ and the notation $f:= f(\lambda_j)$ and $\tilde{f} := f(\tilde{\lambda}_k)$.

\begin{align}
F &= \braket{x}{\tilde{x}} = \frac{1}{\sqrt{p\tilde{p}}} \sum_{j \in \W} |\beta_j|^2 \sum_{k=0}^{T-1} |\alpha_{k|j}|^2 f(\lambda_j) f(\tilde{\lambda}_k) \label{eq:firstline}\\
&= \frac{1}{p \sqrt{1+\frac{\tilde{p}-p}{p}} }\sum_{j\in\W} |\beta_j|^2  \mathbb{E}_{K|J}[ f^2 + (\tilde{f} -f)f]\\
&= \frac{\frac{\sum_{j\in\W} |\beta_j|^2  f^2 }{p} + \frac{\sum_{j\in\W} |\beta_j|^2 \mathbb{E}_{K|J}[(\tilde{f} -f)f]}{p}}{\sqrt{1+\frac{\tilde{p}-p}{p}}} = \frac{1 + \frac{\sum_{j\in\W} |\beta_j|^2 \mathbb{E}_{K|J}[(\tilde{f} -f)f]}{p}}{\sqrt{1+\frac{\tilde{p}-p}{p}}} \\
& \geq \left( 1 + \frac{\sum_{j\in\W} |\beta_j|^2 \mathbb{E}_{K|J}[(\tilde{f} -f)f]}{p} \right) \left(1 - \frac{1}{2} \frac{\tilde{p}-p}{p} \right) \\
&= 1 - \frac{\sum_{j\in\W} |\beta_j|^2 \mathbb{E}_{K|J}[(\tilde{f} -f)^2]}{2p} - \frac{\sum_{j\in\W} |\beta_j|^2 \mathbb{E}_{K|J}[(\tilde{f} -f)f]}{p} \frac{\tilde{p}-p}{2p}. \label{eq:continuefromhere}
\end{align}
Note that in Eq.~\eqref{eq:firstline}, we are allowed to restrict the sum to be only be over the $j\in \W$ as for all other $j$s, $f(\lambda_j)=0$. We also define $\delta:= t_0 (\lambda_j- \tilde{\lambda}_k)$ to make the $t_0$ dependence explicit. The rest of the proof is devoted to upper-bounding each of the two terms being subtracted in Eq.\eqref{eq:continuefromhere}. The first is
\begin{align}
\frac{\sum_{j\in \W} |\beta_j|^2 \mathbb{E}_{K|J} [(\tilde{f}-f)^2]}{2 \sum_{j\in \W} |\beta_j|^2} \leq \frac{\sum_{j\in \W} |\beta_j|^2 \mathbb{E}_{K|J} [\pi^2 \kappa^2 r^2 \delta^2 /t_0^2]}{2 \sum_{j\in \W} |\beta_j|^2} \leq O(\kappa^2 r^2/t_0^2)
\end{align}
where the first inequality follows from Lemma \ref{lem:lipschitz} and the second inequality follows from the fact that $\mathbb{E}[\delta^2] \leq O(1)$ even when conditioned on an arbitrary value of $j$. 

We now turn our attention to the last term of Eq.\eqref{eq:continuefromhere} and we will bound each of the two terms in the product. We first bound the first term in the product:
\begin{align}
\frac{\sum_{j\in\W} |\beta_j|^2 \mathbb{E}_{K|J}[(\tilde{f} -f)f]}{\sum_{j\in \W} |\beta_j|^2} &\leq \frac{\sum_{j\in\W} |\beta_j|^2 \mathbb{E}_{K|J}[|\tilde{f} -f| |f|]}{\sum_{j\in \W} |\beta_j|^2} \\
&\leq \frac{\pi \kappa r}{t_0} \frac{\sum_{j\in\W} |\beta_j|^2 \mathbb{E}_{K|J}[|\delta| |f|]}{\sum_{j\in \W} |\beta_j|^2} \leq O( r \kappa/t_0)\label{eq:first_term}
\end{align}
where the second inequality follows from Lemma \ref{lem:lipschitz} and the last inequality follows from the fact that $|f|<1$ and $\mathbb{E}_{K|J}[|\delta|]<O(1)$ even when conditioned on an arbitrary value of $j$. The second term of the product is
\begin{align}
\frac{\tilde{p}-p}{2p}  &= \frac{1}{2p}\sum_{j=0}^{N-1} |\beta_j|^2 \mathbb{E}_{K|J}[\tilde{f}^2 -f^2] \\
&=  \frac{1}{2p} \left(\sum_{j \in \W} |\beta_j|^2 \left(2\mathbb{E}_{K|J}[(\tilde{f} -f)f] + \mathbb{E}_{K|J}[(\tilde{f} -f)^2] \right) + \right.\\
&\left.\sum_{j \not \in \W} |\beta_j|^2 \left(2\mathbb{E}_{K|J}[(\tilde{f} -f)f] + \mathbb{E}_{K|J}[(\tilde{f} -f)^2] \right) \right)\\
&\leq O(r\kappa/t_0) + O(r^2\kappa^2/t_0^2) + \frac{\sum_{j \not \in \W} |\beta_j|^2 (2\mathbb{E}_{K|J}[(\tilde{f} -f)f] + \mathbb{E}_{K|J}[(\tilde{f} -f)^2]  )}{2\sum_{j \in \W} |\beta_j|^2} \\
&\leq O(r\kappa/t_0) + O(r^2\kappa^2/t_0^2)  + \frac{\sum_{j \not \in \W} |\beta_j|^2 (2 \pi \kappa r \mathbb{E}_{K|J}[|\delta|]/t_0 + \pi^2 \kappa^2 r^2 /t_0^2 \mathbb{E}_{K|J}[\delta^2]  )}{2\sum_{j \in \W} |\beta_j|^2} \\
& \leq O(r\kappa/t_0) + O(r^2\kappa^2/t_0^2) + \frac{1-p}{p} O(r \kappa/t_0) + \frac{1-p}{p}O(r^2 \kappa^2/t_0^2) \\
& \leq O\left(\frac{r\kappa}{c t_0}\right)+O\left(\frac{r^2\kappa^2}{c t_0^2}\right).
\label{eq:boundonp_relative}
\end{align}
where the last inequality follows from the theorem's assumption that $\lVert \Pi_{{\rm row}(F^{\dag})} \ket{\psi} \rVert_2 > c$. \footnote{This parameter appears implicitly in \cite{HHL} also when they claim their initial probability of success before postselection to be $O(\kappa)$, where the $O$ hides the $c$.}Keeping only highest-order terms in $r \kappa/t_0$ and multiplying this with Eq~\eqref{eq:first_term} yields that the second term of Eq.\eqref{eq:continuefromhere} is of $O\left(\frac{r^2\kappa^2}{c t_0^2}\right)$. Thus overall we have that $F = \langle x | \tilde{x} \rangle \geq 1-O(\frac{r^2 \kappa^2}{c t_0^2}) $ which implies that 
\be
\lVert \ket{\tilde{x} } - \ket{x} \rVert = \sqrt{2(1-\text{Re} \braket{x}{\tilde{x}} )} \leq O\left(\frac{r \kappa}{\sqrt{c} t_0}\right)
\ee
as claimed.
\end{proof}

\subsubsection{LMR protocol for Density Matrix Exponentiation}\label{appsubsubsec:error_LMR}
We present the following Lemma, which gives the guarantees for the LMR protocol for density matrix exponentiation described in \cite{lloyd2013QPrincipalCompAnal, kimmel2016hamiltonian}.
\begin{lemma}[Hamiltonian simulation via density matrix exponentiation \cite{lloyd2013QPrincipalCompAnal, kimmel2016hamiltonian}] \label{lem:DME}
For two density matrices $\rho, \tilde{\rho} \in D(\mathcal{H}_{\mathrm{C}})$. 
Let $\sigma \in \mathrm{D}\left(\mathcal{H}_{\mathrm{C}} \otimes \mathcal{H}_{\mathrm{B}}\right)$ be an unknown quantum state and let $t \in \mathbb{R}$. Then there exists a quantum algorithm, given copies of $\rho, \tilde{\rho}$, that transforms $\sigma_{\mathrm{CB}} \otimes (\rho_{\mathrm{C}_{1}} \otimes \tilde{\rho}_{\mathrm{C}'_{1}}) \cdots \otimes (\rho_{\mathrm{C}_{n}} \otimes \tilde{\rho}_{\mathrm{C}'_{n}})$ into $\tilde{\sigma}_{\mathrm{CB}}$ such that
\begin{equation}\label{eq:W}
\frac{1}{2}\left\lVert (e^{-i (\rho -\tilde\rho) t} \otimes \id_\mathrm{B}) \sigma_{\mathrm{CB}} (e^{i (\rho -\tilde\rho) t} \otimes \id_\mathrm{B})-\tilde{\sigma}_{\mathrm{AB}}\right\rVert_{1} \leq \delta
\end{equation}
In other words, this quantum algorithm implements the unitary $e^{-i (\rho -\tilde\rho) t}$ on state $\sigma$ up to error $\delta$ in diamond norm, using $O\left(t^{2} / \delta\right)$ copies of $\rho,\tilde{\rho}$. The same number of copies is needed if the desired operation is instead the controlled version of the unitary, $\ketbra{0}{0}\otimes \id_C + \ketbra{1}{1} \otimes e^{-i (\rho -\tilde\rho) t}$. 
\end{lemma}
For any $t \in \mathbb{R}$,  the LMR protocol transforms $\sigma_{\mathrm{CB}} \otimes \rho_{\mathrm{C}_{1}} \otimes \cdots \otimes \rho_{\mathrm{C}_{n}}$ into $\tilde{\sigma}_{\mathrm{CB}}$ such that
\begin{equation}\label{eq:LMR}
\frac{1}{2}\left\lVert (e^{-i \rho_{\mathrm{C}} t} \otimes \id_B) \sigma_{\mathrm{CB}} (e^{i \rho_{\mathrm{C}} t} \otimes \id_B)-\tilde{\sigma}_{\mathrm{CB}}\right\rVert_{1} \leq \delta
\end{equation}
with $n=O\left(t^{2} / \delta\right)$ copies of $\rho$. 
This is achieved as follows. Let $S_{CC_k} \in \mathrm{U}\left(\mathcal{H}_{\mathrm{C}} \otimes \mathcal{H}_{\mathrm{C}_{k}}\right)$ be the unitary operator that swaps systems $C$ and $C_{k},$ i.e. $S_{CC_k} |i\rangle_{\mathrm{C}}|j\rangle_{\mathrm{C}_{k}}=|j\rangle_{\mathrm{C}}|i\rangle_{\mathrm{C}_{k}}$ for all $i, j \in$ $\left\{1, \ldots, \operatorname{dim}\left(\mathcal{H}_{\mathrm{C}}\right)\right\} .$ The protocol applies the exponentiated swap gate $e^{-i \Delta t S_{CC_k}}$ between systems $C$ and $C_k$ for time $\Delta t$ (to be chosen), and then discards system $C_k$. This is done in succession for $k$ ranging from $1$ to $n$. Performing the above procedure for a single value of $k$ simulates $e^{-i\rho\Delta t}$ up to an error of  $(\Delta t)^2$. Doing so for $n$ values of $k$ in succession, therefore, simulates $e^{-i\rho n\Delta t}$, accruing a total error of $n(\Delta t)^2$. For a given value of $t, \delta$, we may therefore simulate $e^{-i\rho t}$ up to error $\delta$ by using the LMR protocol with $n = O(t^2/\delta)$ as stated.
Though we have given its complexity in terms of $n$, the number of copies of the density matrix to be exponentiated $\rho$, this is also the gate complexity up to polylogarithmic factors. This is because each copy of $\rho$ is acted upon by exactly one exponentiated swap gate generated by $S_{CC_k}$. The exponentiated swap, being generated by a $1$-sparse matrix, can be simulated to accuracy $\epsilon_0$ with gate complexity of $\Ord{\log D \log \frac{1}{\epsilon_0}}$  \cite{low2016HamSimQubitization}.

\subsubsection{Error incurred due to Density Matrix Exponentiation}\label{appsubsubsec:error_DME}

Using quantum phase estimation to estimate the eigenvalues of a unitary W requires the ability to apply controlled versions of $W, W^2, \ldots W^{T}$ on the input state, conditioned on an ancillary state. For Algorithm \ref{algo:DME}, $W$ is enacted using the state preparation unitaries, which incurs some error. The purpose of this section is to analyze the gate complexity of implementing the $W$'s in a single iteration of QPE up to a total error of $\eps_{D}$. We will compare two versions of finite-time quantum phase estimation:
\begin{itemize}
\item $\tilde{P}(\tilde{D})$, which is quantum phase estimation with an imprecise density matrix exponentiation step.
\item $\tilde{P}(D)$, which is quantum phase estimation with a precise density matrix exponentiation step.
\end{itemize}
\begin{lemma}\label{lem:DME_in_QPE}
Given copies of $\rho, \tilde{\rho} \in D(\mathcal{H})$ with $\rho - \tilde \rho =:H$, quantum phase estimation with unitary $W=e^{i(H+\id) t_0/{2T}}$ can be carried out to error $\eps_D$, that is,
\be
\lVert \tilde{P}(\tilde{D}) - \tilde{P}(D) \rVert < \eps_D,
\ee
with a gate complexity of $O(t_0^2/\eps_D ).$
\end{lemma}

For quantum phase estimation, we will need to implement controlled versions of the unitaries
\begin{equation}\label{eq:qpeunitaries}
W^{k} = e^{i (H+\id) t_0 k /2T} \qquad \forall k \in \{2^0,\ldots 2^{log_2 T -1}\}
\end{equation}
where $t_0$ and $T$ are precision parameters that are fixed by the overall desired error. Enacting the controlled version of $W^k$ requires the same gate complexity as enacting its uncontrolled version, which we now analyze. 

As elaborated in Section \ref{appsubsec:DME_procrustes}, $H$ is the difference of two density matrices $\rho, \tilde{\rho}$ which are prepared using the state preparation unitaries. Therefore, we may use density matrix exponentiation (Lemma \ref{lem:DME}) to implement a unitary of the form $e^{i H t_0 k /2T}$ up to error $\delta$ with gate complexity $O\left(\frac{ t_0^2 k^2}{T^2 \delta}\right)$. Eq.\eqref{eq:qpeunitaries} requires us to go one step further and implement an additional identity (scaled appropriately) in the exponent. This is achievable simply by applying a controlled phase, where the control is the $\ket{\tau}$ register. The cost of this is a single additional gate for every application of $W^k$, which adds $\log(T)$ to the overall gate complexity and is suppressed by our big-O notation. 

In the application of $W^{2^i}$, we will choose $k = 2^i$ and $\delta = (\eps_{D}/T)\cdot 2^i$. Since we have established that $i\in\{0,\ldots, \log_2(T) - 1\}$, the total error is
\be
\sum_{i=0}^{\log T-1} \frac{\eps_D}{T} 2^i = \eps_D
\ee
as desired, and the total gate complexity is the sum of the gate complexities of each of the components,
\be \label{eq:DME_error}
\sum_{i=0}^{\log T-1} O\left(\frac{t_0^2 2^{2i}}{T \eps_D \cdot 2^i}\right) = O\left(\frac{t_0^2}{\eps_D}\right).
\ee

\subsubsection{Overall error and gate complexity}\label{appsubsubsec:proof_DME}

To aid the analysis of the overall error, we define some versions of the desired final state with errors in successively fewer components of the algorithm, using the overscript $\tilde{}$ to denote inexact application of the corresponding step.
\begin{enumerate}
\item $\ket{x_{\tilde{P}(\tilde{D})}}$, the state after post-selection with inexact DME and inexact QPE. This is the state that our algorithm produces.
\item $\ket{x_{\tilde{P}(D)}}$ (previously known as $\ket{\tilde{x}}$ in Section \ref{appsubsubsec:error_QPE}), the state after post-selection with exact DME and inexact QPE.
\item $\ket{x_{P(D)}}$ (previously known as $\ket{x}$ in Section \ref{appsubsubsec:error_QPE}), the state after post-selection with exact DME and exact QPE. This is the ideal final state.
\end{enumerate}

Here, post-selection means amplitude amplification followed by measurement. We will now combine the error analysis of the individual components to prove Theorem \ref{thm:QSVT_DME}. 
\begin{proof}[Proof of Theorem \ref{thm:QSVT_DME}]
We will bound the accumulated errors of the entire procedure in the {\em post-selected} state, that is:
\begin{equation}
   \lVert \ket{x_{P(D)}} - \ket{x_{\tilde{P}(\tilde{D})}}\rVert \leq   \underbrace{\lVert\ket{x_{\tilde{P}(\tilde{D})} }- \ket{x_{\tilde{P}(D)}}\rVert}_{\text{error of QPE due to imperfect DME}} + \underbrace{\lVert \ket{x_{\tilde{P}(D)}} - \ket{x_{P(D)}} \rVert}_{\text{inherent error of QPE}} \leq \varepsilon.
\end{equation}
The rest of this proof proceeds as follows: to handle the second term, Theorem \ref{thm:QPE_DME}.2 will determine the value of $t_0$ (which sets the precision of phase estimation) so that the inherent error of QPE after postselection is $\eps/2$. To handle the first term, we will combine the value of $t_0$ obtained earlier and an analysis of amplitude amplification, to compute the choice of $\eps_D$ in Lemma \ref{lem:DME_in_QPE} that ensures the first term is $\eps/2$. This then determines the overall gate complexity.

Theorem \ref{thm:QPE_DME} immediately yields the bound on the second term (inherent error of QPE) with the choice $t_0 = O(\frac{r\kappa}{\sqrt{c}\eps})$. We now handle the first term. We first claim that the number of repetitions of $U_{\rm all}$ (see Eq.~\eqref{eq:unitary}) required to perform amplitude amplification is $1/\sqrt{c}$. To see this, note that a re-writing of Eq.\eqref{eq:boundonp_relative} yields that 
\be
|\tilde{p}-p|\leq p \cdot O\left(\frac{r \kappa}{pt_0}\right).
\ee
With our choice of $t_0=O(\frac{r\kappa}{\sqrt{c}\eps})$, the right-hand-side is $O(\eps \sqrt{c})$ and if we choose $\eps = O(c^{3/2})$, we have
\be
\tilde{p} \overset{O(c^2)}{\approx} p := \sum_{j: \lambda_j \geq 1/\kappa} |\beta_j|^2 \geq c.
\ee
$\tilde{p}$ is, however, also the probability of success (i.e. projecting onto $\ket{\rm{well}}$) if one were to measure the state $\ket{\tilde{\varphi}}$ in Eq.~\eqref{eq:states_2} before amplitude amplification. Therefore, the number of repetitions of $\tilde{P}$ to achieve a constant error is at most $1/\sqrt{c}$ \cite{brassard2002AmpAndEst}. This is also the factor amplifying the error-due-to-DME within a single round of $\tilde{P}$. Therefore in Lemma \ref{lem:DME_in_QPE}, we need to choose $\eps_D = \sqrt{c} \eps/2.$ \footnote{Technically, there is one last source of error we need to bound which is the error due to the final measurement on the amplified state. Since the amplified state already has a large constant amplitude $(>8/\pi^2)$ in the subspace onto which the measurement projects, however, this only contributes a constant factor to the overall error which is suppressed by our big-O notation.} Finally, substituting in $t_0 = O(\frac{r \kappa}{\eps \sqrt{c}})$ and $\eps_D =\sqrt{c}\eps/2$ into Lemma \ref{lem:DME_in_QPE} yields the claimed gate complexity of 
\begin{equation}
O\left(\frac{r^2 \kappa^2}{c^{3/2}\eps^3}\right)
\end{equation}
and this proves Theorem \ref{thm:QSVT_DME}.
\end{proof}
\end{document}